\documentclass[9pt,usenames,dvipsnames]{article} 
\usepackage{amsmath,amssymb} 
\usepackage{caption}
\usepackage{placeins}
\usepackage[usenames]{color}
\usepackage{bm}
\usepackage[draft]{changes}
\usepackage{ying}

\usepackage{comment}
\usepackage{multirow}
\usepackage{rotating}
\usepackage{fullpage}
\usepackage{authblk}
\usepackage{enumerate}
\usepackage{geometry}
\usepackage{float}
\usepackage{subcaption}
\usepackage{tikz,tabularx}
\usetikzlibrary{patterns}
\usetikzlibrary{arrows}
\usetikzlibrary{tikzmark, positioning, fit, shapes.misc}

\usetikzlibrary{decorations.pathreplacing, calc}
\tikzset{brace/.style={decorate, decoration={brace}},
  brace mirrored/.style={decorate, decoration={brace,mirror}},
}

\newcolumntype{g}{>{\columncolor{red}}c}

\usepackage{graphicx,amssymb}
\usepackage{xcolor}

\usepackage[colorlinks,
            linkcolor=red,
            anchorcolor=blue,
            citecolor=blue
            ]{hyperref}
\usepackage{algorithm}
\usepackage{algorithmic}

\newcommand{\cs}{\textnormal{CS}}
\newcommand{\oracle}{*}

\let\hat\widehat
\let\tilde\widetilde

\def \calib {\textrm{calib}}

\def \test {\textrm{test}}

\def \quant {\textrm{Quantile}}

\def \fdp {\textnormal{FDP}}
\def \fcp {\textnormal{FCP}}
\def \ks {\textrm{KS}}
\def \hc {\textrm{HC}}

\def \bj {\textrm{BJ}}

\theoremstyle{plain}

\usepackage{hyperref}
\def\@#1\@{\begin{align}#1\end{align}}
\def\$#1\${\begin{align*}#1\end{align*}}

\usepackage{xcolor}

\usepackage{color-edits}
\addauthor[Ying]{ying}{magenta}

\title{Everywhere Valid Bounds on False Discovery Proportions\\ in Conformal Inference}

\author[1]{Ziang Song}
\author[2]{Ying Jin}
\author[1,3]{Emmanuel J. Cand\`es}

\affil[1]{Department of Statistics, Stanford University}
\affil[2]{Department of Statistics and Data Science, University of Pennsylvania}
\affil[3]{Department of Mathematics, Stanford University}

\date{}

\begin{document}

\maketitle                                                                                                                                                                                                                                                                                                                                                   




\begin{abstract}
Modern applications of conformal inference to multiple testing problems---such as outlier detection and candidate selection---often involve selecting test samples whose conformal p-values fall below a threshold. 
The quality of such methods is often measured by the false discovery proportion (FDP), defined as the fraction of incorrect selections.
Existing approaches typically control the expected value of the FDP, using methods such as the Benjamini–Hochberg procedure. This approach fails to provide high-probability bounds on the realized false discovery proportion and invalidates statistical guarantees if the rejection threshold is selected {\em after} inspecting the data. This paper establishes finite-sample, distribution-free upper bounds on the FDP that hold simultaneously over all possible rejection thresholds, enabling arbitrary post hoc selection of the threshold. Simultaneous validity is achieved by constructing a high-probability envelope for the empirical distribution function of `null' conformal p-values by sampling from their joint distribution. Furthermore, our framework allows practitioners to modulate the envelope's shape, thereby producing tight bounds in rejection regions of primary interest. We use this flexible approach to derive simultaneous FDP upper bounds for both outlier detection and conformal selection. We demonstrate through synthetic and real-data experiments that the resulting bounds are both valid and substantially less conservative than those derived from existing approaches.
\end{abstract}

    





\section{Introduction}
Conformal p-values have  become a fundamental tool for distribution-free hypothesis testing~\citep{vovk2005algorithmic,bates2023testing}. Their appeal is that they can be constructed using arbitrary predictive models while offering finite-sample guarantees. These properties have enabled broad applications ranging from outlier detection \citep{bates2023testing, marandon2024adaptive} to  scientific discovery~\citep{jin2023aselection,jin2023bmodel,jin2026txconformal}, and decision-making in which practitioners wish to select a subset of test samples obeying a user-specified criterion. 

A common strategy in these applications is to select test samples whose p-values (or scores) fall below a threshold. 
Such a selection rule can be viewed as a binary classifier that predicts whether each test sample is of interest. Thresholding decision rules naturally induce a family of binary decisions; as the threshold varies, the resulting decisions trace out curves of binary error metrics across the test samples, analogous to precision-recall curves in classification~\citep{davis2006relationship,powers2020evaluation}.

The reliability of such selections is naturally quantified by the False Discovery Proportion (FDP) \citep{benjamini1995controlling}, defined as the ratio of false discoveries to the total number of selections. In classification where we would like to single out samples with a label of interest, the FDP would be the fraction of predicted positives that are incorrect. This is simply one minus the precision. 
Controlling the FDP/precision is often crucial because each selected unit typically triggers costly follow-up actions, such as experimental validation, manual inspection, or downstream scientific investigation. We illustrate this through two motivating examples. 
\paragraph{Outlier Detection} In outlier detection \citep{bates2023testing, marandon2024adaptive}, one has access to a calibration set $\cD_\calib = \{X_i\}_{i = 1}^n$ where each $X_i$ is independently and identically distributed (i.i.d.) from an unknown distribution $\cP_0$. Given test data $\cD_\test = \{X_{n+j}\}_{j=1}^m$, the goal is to detect which points are likely to be outliers, i.e., drawn from a  distribution different from $\cP_0$. This task arises in fraud detection, anomaly detection in manufacturing, and the detection of unusual patterns in scientific data. The conformal inference framework provides a principled approach by calculating p-values $\{p_j\}_{j=1}^m$ that quantify how well each test sample conforms to the calibration data; smaller values indicate stronger evidence of being an outlier. Common  procedures~\citep{bates2023testing, marandon2024adaptive} use the thresholding rule $\cR(t) = \{j: p_j \le t\}$ in which $t\in [0,1]$
is a data-driven threshold computed via the Benjamini-Hochberg (BH) procedure~\citep{benjamini1995controlling}.  Let $\cH_0 = \{j: X_{n+j} \sim \cP_0\}$ denote the index set of inliers so that a false discovery means a selection from $\cH_0$. The FDP is then 
\@\label{eq:def_fdp_t}
\fdp(t) = \frac{|\cR(t) \cap \cH_0|}{\max\{1, |\cR(t)|\}}.
\@
The numerator above $|\cR(t) \cap \cH_0|$ is the number of inliers incorrectly selected as outliers. 
Maintaining a low FDP is often desired since the identified outliers often call for follow-up investigations, and a low FDP indicates that these costly investigations are rarely wasted~\citep{benjamini1995controlling,efron2012large}.


\paragraph{Drug Discovery} Conformal p-values also arise in drug discovery pipelines~\citep{jin2023aselection, jin2023bmodel}, where early-stage screening aims to identify promising drug candidates from a large library.  
Machine learning predictions are useful in an initial screening stage to shortlist candidates; in later, more costly stages, only those shortlisted are carefully investigated to confirm promising cases. In this problem, the calibration set is $\cD_{\calib} = \{(X_i, Y_i)\}_{i=1}^n$ where $X_i$ is the biochemical feature of the $i$th drug, and $Y_i$ is the experimentally evaluated property, such as the binding affinity to a disease target. For new test drugs $\{X_{n+j}\}_{j=1}^m$  with unknown properties $\{Y_{n+j}\}_{j=1}^m$, the goal is to find those $Y_{n+j}$'s that exceed a pre-specified threshold $c_{n+j}$. 

The conformal selection framework \cite{jin2023aselection} constructs conformal p-values to make such selections. Here, a smaller p-value $p_j$ provides strong evidence for a large outcome. The selection set is given by $\cR(t) = \{j: p_j \le t\}$, where $t\in \RR$ is determined by the BH procedure at a desired nominal level~\citep{benjamini1995controlling}. A false discovery is here a selected  drug  whose response $Y_{n+j}$ is no larger than $c_{n+j}$, and thus the FDP is as in~\eqref{eq:def_fdp_t} with $\cH_0 = \{j: Y_{n+j} \le c_{n+j}\}$. Controlling the FDP ensures a high precision or `hit rate', which is critical to prioritizing experimental validation efforts.

\vspace{-0.5em}
\paragraph{Limitations of existing methods} 
Existing methods mainly focus on controlling the false discovery rate (FDR), which is the expected value of the FDP averaged over the randomness in the data. For a pre-specified error rate $\alpha\in(0,1)$, these methods determine a data-dependent threshold $t(\alpha)$  such that $\EE[\fdp(t(\alpha))]\leq \alpha$. These approaches suffer from two major limitations: 
\begin{itemize}
\item \emph{What's the realized FDP?} The realized FDP on a particular dataset may substantially exceed its expected value, even when the FDR is controlled \citep{kluger2024central}. FDR control only guarantees that the FDP is small \emph{on average} over repeated realizations of the data, and thus provides no guarantee for the dataset at hand. However, practitioners often desire guarantees for the \emph{data at hand}. They might want to know that with 90\% or 95\% chance, the FDP is below 20\%. It may be of little comfort to know that over independent realizations of the data, the FDR may be at most 10\% since this does not say anything about {\em their} study.  

    \item \emph{Fixed-$\alpha$ guarantees only.} These methods require the nominal level $\alpha$ to be fixed in advance, which limits their flexibility. A practitioner may wish to choose $\alpha$ adaptively after inspecting the data. For example, if $\alpha = 0.05$ yields very few discoveries, they may want to increase it to $0.1$ to include more discoveries. Such post hoc adjustments, however, invalidate FDR control, see Figure~\ref{fig:adaptive_example} in the Appendix for a concrete illustration.
\end{itemize} 
 
In this work, we address these limitations by proposing a framework to construct simultaneously valid, high-probability upper bounds on the FDP that remain valid for all thresholds $t\in [0,1]$. This enables flexible, data-driven exploration while maintaining rigorous statistical guarantees. Given a confidence level $\delta\in (0,1)$, we provide a data-driven function $\hat{\text{FDP}}(t)$ obeying 
$$
\PP\Big( \fdp(t) \leq \hat{\fdp}(t),~\forall t\in [0,1]\Big) \geq 1- \delta.
$$
That is, with probability 90\%, say, the observed curve $\hat{\fdp}(t)$ everywhere dominates the unknown FDP curve $\fdp(t)$.

Such simultaneous FDP bounds have been extensively studied in multiple testing, but existing approaches are not tailored to the finite-sample dependence structure of conformal p-values. Our approach instead leverages the exact joint distribution of null conformal p-values induced by exchangeability.
Figure \ref{fig:intro} previews our results on a real drug-target interaction task, with the full experimental setup deferred to Section \ref{sec:cs-experiment}. Briefly, the goal is to select drug-target pairs whose binding affinity exceeds a given threshold. The figure demonstrates that our upper bounds consistently lie above the true FDP across all selection sets, ensuring valid error control regardless of the number of selections. This not only offers practitioners reliable, instance-wise control of FDP but also allows them to adaptively vary the selection threshold while maintaining rigorous guarantees.

\begin{figure}
    \centering
    \begin{subfigure}{0.5\linewidth}
        \includegraphics[width=\linewidth]{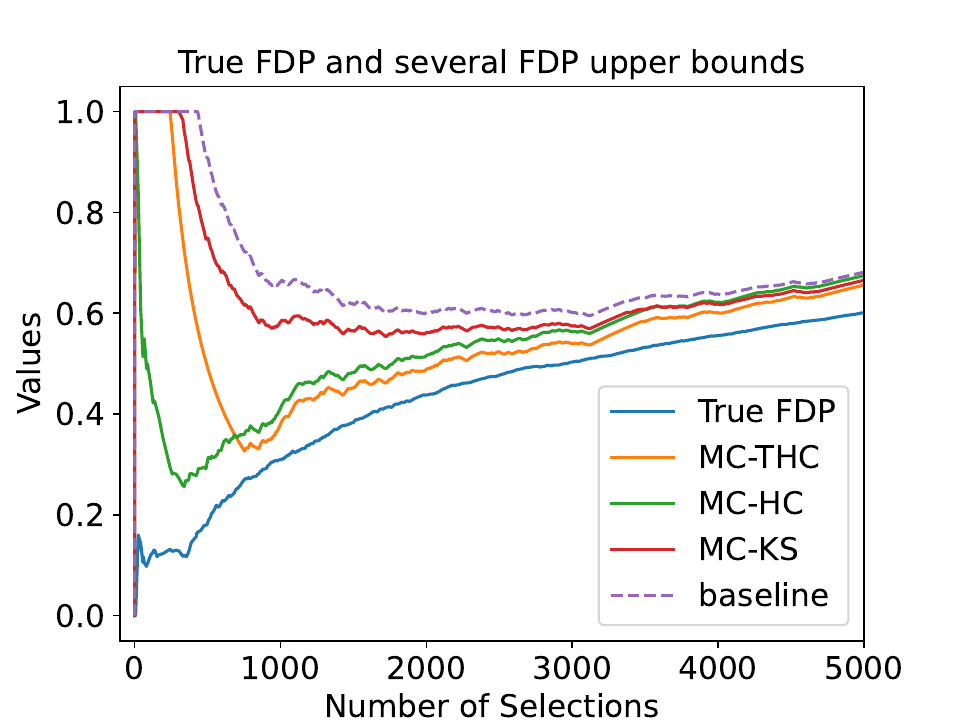}
        \caption{Single realization}
        \label{fig:intro_single}
    \end{subfigure}
    \hfill
    \begin{subfigure}{0.45\linewidth}
        \includegraphics[width=\linewidth]{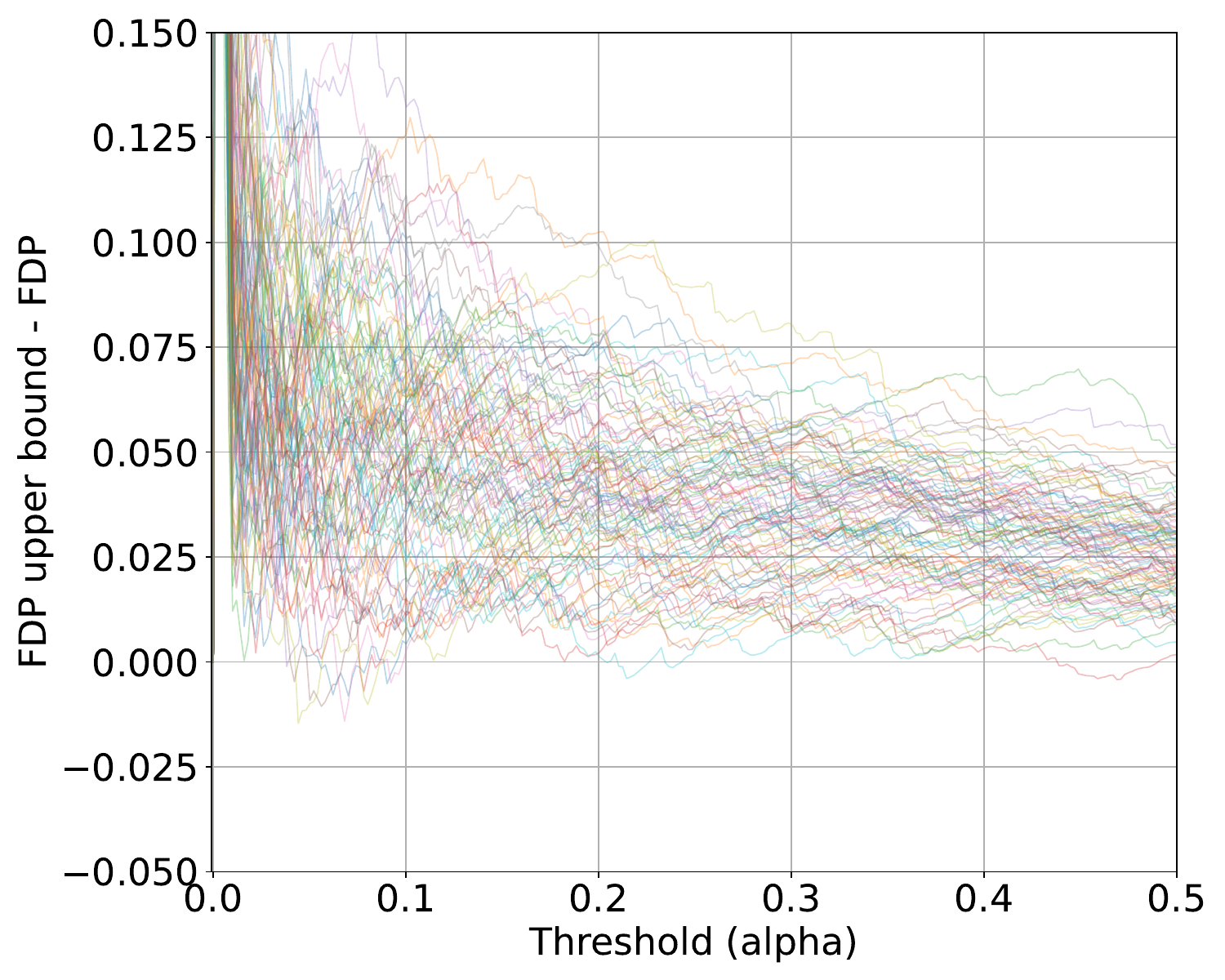}
        \caption{Residuals over 100 experiments}
        \label{fig:intro_multiple}
    \end{subfigure}
    \caption{\textbf{Simultaneous FDP bounds in a drug-target interaction task.}
\textbf{(a)} One realization of the true FDP (blue) and simultaneous upper bounds constructed by our method  with the following statistics: Truncated Higher Criticism (MC-THC), Higher Criticism (MC-HC), and Kolmogorov-Smirnov (MC-KS). The dashed line is the upper bound adapted from \citep{gazin2024transductive}. 
\textbf{(b)} Residuals (upper bound minus true FDP) across 100 independent experiments. Positive values indicate valid coverage. The $x$-axis is the selection threshold $t \in [0,1]$. 
   }
    \label{fig:intro}
\end{figure}

Our main contributions are threefold:
\begin{itemize}
    \item We establish a novel high-probability upper bound on the empirical cumulative distribution function (CDF) of conformal uniform variables; see Section \ref{sec:upper-bound-oracle}. Empirically, our bound is tighter than existing bounds based on DKW-type concentration inequalities \cite{gazin2024transductive}. As a direct consequence, we obtain simultaneous upper bounds on the false coverage proportion in conformal prediction that are substantially less conservative than existing approaches; see Section \ref{sec:fcp_bounds}. 

    \item  
    For outlier detection, we develop an upper bound for the FDP of selection sets holding at all rejection thresholds simultaneously; see Section \ref{sec:outlier-detection}. Our approach combines our improved bound on the empirical CDF with two different tightening techniques from \citep{hemerik2019permutation} and \citep{gazin2024transductive}, respectively. Through extensive simulation studies, we demonstrate the practical effectiveness of our quantitative results. 
    
    \item For conformal selection, we establish, to the best of our knowledge, the first finite-sample FDP bounds that hold uniformly over all rejection thresholds; see Section \ref{sec:conformal-selection}. We validate our theoretical results on the DAVIS dataset \citep{davis2011comprehensive}, and show that our framework allows flexible exploration while providing rigorous error control guarantees. 
\end{itemize}

Underlying all this is a general framework for deriving high-probability upper bounds on random functions via sampling from the joint distribution of null conformal p-values. Compared with approaches based on concentration inequalities, this approach yields substantially sharper bounds while preserving finite-sample validity. 

\paragraph{Code availability.}
Code for reproducing the numerical experiments and figures is available at
\url{https://github.com/sza919/everywhere-valid-fdp-bounds-in-conformal-inference}.




\section{Related work}
Our work lies at the intersection of four lines of literature: conformal inference for multiple testing,  false coverage proportion control in conformal prediction, post-hoc simultaneous FDP inference, and Monte Carlo/global-envelope calibration.

\paragraph{Conformal inference for multiple testing.}
A growing line of work uses conformal inference to construct finite-sample valid $p$-values for multiple testing problems. \cite{bates2023testing} introduced conformal $p$-values for outlier detection and showed how they can be combined with multiple-testing procedures to control the false discovery rate. This line has since been extended to more powerful and adaptive outlier or novelty detection methods, including integrative conformal $p$-values using labeled outliers or side information \citep{liang2022integrative}, adaptive novelty detection with FDR guarantees \citep{marandon2024adaptive}, and full-conformal novelty detection \citep{lee2025full}. In a related direction, \cite{jin2023aselection} introduced conformal selection, where conformal $p$-values are used to select units whose unobserved outcomes exceed user-specified thresholds while controlling the false discovery rate. Subsequent work has extended conformal selection to covariate shift via weighted conformal $p$-values, optimized and adaptive selection procedures, multivariate responses, and post-hoc selection based on $e$-variables \citep{jin2023bmodel, bai2024optimized,gui2025acs,bai2025multivariate,zhu2026beyond}. These works provide finite-sample FDR-type guarantees for conformal outlier detection and conformal selection, typically for a pre-specified testing or selection rule. Our work is complementary: we provide high-probability upper bounds on the realized false discovery proportion, uniformly over all rejection thresholds, thereby allowing post-hoc threshold selection.


\paragraph{False coverage proportion control and exact conformal p-value laws.}
Our work is also closely related to recent work on simultaneous coverage errors in conformal prediction. Rather than controlling only the marginal coverage probability of each prediction set, this literature studies the realized fraction of test points whose prediction sets fail to cover, often called the false coverage proportion (FCP). \citep{marques2025universal} derived the exact finite-sample distribution of the empirical coverage of split conformal prediction sets for a batch of future observations. In the transductive setting, \citep{gazin2024transductive} characterized the joint distribution of conformal $p$-values through a P\'olya-urn representation and derived concentration inequalities for their empirical distribution function. Related asymptotic results further characterize the fluctuations of coverage errors in conformal prediction \citep{gazin2024asymptotics}. Building on this distributional viewpoint, \citep{blain2025false} proposed CoJER, a procedure that controls the FCP in probability using the transductive conformal $p$-value distribution. These papers and ours share the same starting point: the joint distribution of conformal $p$-values contains substantially more information than marginal validity alone. Existing simultaneous FCP guarantees based on this viewpoint are asymptotic or rely on concentration bounds, whereas our procedure samples directly from the finite-sample conformal $p$-value distribution to obtain simultaneous bounds that remain valid for any finite number of simulations. In addition, our inferential target is the false discovery proportion in conformal
multiple testing, rather than the false coverage proportion of conformal prediction sets.

\paragraph{Post-hoc simultaneous FDP bounds.}
Our objective is rooted in the literature on simultaneous false discovery proportion bounds. Beyond classical FWER and FDR control, earlier work studied tail-probability guarantees such as $k$-FWER control and exceedance control of the FDP \citep{genovese2006exceedance,romano2007control}. These guarantees can be viewed as precursors to modern simultaneous FDP inference, where the goal is to provide a high-probability bound for the realized FDP uniformly over a family of rejection sets. A prominent formulation is the joint error rate (JER) framework \cite{blanchard2020post}, which provides post-hoc confidence bounds on the number of false positives. In several settings, such guarantees are implemented through a critical vector or envelope: one constructs simultaneous upper bounds on the number of false discoveries along an ordered rejection path and then extends these bounds to obtain FDP bounds for data-dependent rejection sets. This perspective appears in permutation-based simultaneous FDP bounds \citep{meinshausen2006false,hemerik2019permutation} and is closely related to closed-testing and Simes-type approaches \citep{goeman2011multiple,goeman2019simultaneous,goeman2021only,vesely2023permutation}, reference-family methods \citep{blanchard2020post}, and bounds for structured, regression, and online settings \citep{katsevich2020simultaneous}. Closest to our conformal setting, \citep{gazin2024transductive} also use the exact null conformal $p$-value distribution to derive uniform in-probability guarantees, including FDP-type bounds for novelty detection. Our work differs in the calibration step: rather than relying on concentration inequalities, we sample directly from its exact finite-sample law to construct high-probability envelopes. This yields sharper finite-sample, distribution-free simultaneous FDP bounds tailored to conformal outlier detection and conformal selection.

\paragraph{Monte Carlo calibration and global envelopes.}
Methodologically, our construction is connected to exact Monte Carlo tests and global envelope methods. Classical Monte Carlo tests calibrate an observed statistic by simulating its distribution under the null and comparing the observed value with its simulated counterparts \citep{dwass1957modified,hope1968simplified,besag1991sequential,dufour2006monte}. Global envelope tests extend this idea from scalar statistics to functional statistics by constructing simultaneous envelopes over an index set \citep{myllymaki2017global}. Our method can be viewed as a conformal analogue of such a global envelope construction. The indexed object is the empirical distribution function of null conformal $p$-values, the index is the rejection threshold, and the reference distribution is the exact finite-sample joint law of conformal $p$-values. Sampling from this law allows us to build an envelope for the null empirical process, which then translates into a simultaneous FDP upper bound. Thus, the novelty is not Monte Carlo calibration alone, but its combination with the exact conformal $p$-value distribution to obtain finite-sample, distribution-free FDP certificates.

\section{Basic properties of conformal uniform variables}

\subsection{Conformal uniform variables}
Following the split conformal prediction framework~\citep{vovk2005algorithmic}, we assume access to a fully observed calibration set $\cD_\calib = \{Z_1, Z_2, \ldots, Z_n\}$ and a set of partially observed test units $\cD_\test = \{Z_{n+1}, Z_{n+2}, \ldots, Z_{n+m}\}$. Each $Z_{i}\in \mathcal{Z}$ may represent either a feature-response pair or a feature vector alone depending on the specific context. The conformity score, denoted as $s(\cdot)$, is a function that assigns a real number to each data point that measures its ``conformity'' with the rest of the dataset. We assume this score function is obtained by a training process independent of both the calibration and test data. For example, in outlier detection, $s(Z_i)$ may be a one-class SVM classifier~\citep{scholkopf1999support} assigning a high score to $Z_i$ if $Z_i$ is not likely to be an outlier and low score otherwise.
To simplify notation and discussion, throughout, we assume no ties among scores:

\begin{assumption}[No Tie]\label{asmp:no-tie}
    $\mathbb{P}(s(Z_i) = s(Z_j)) = 0$ for all $i \neq j$.
\end{assumption}
We rely on the standard exchangeability assumption to characterize the null distribution:
\begin{assumption}[Exchangeability]\label{asmp:exchangeability}
The combined set of calibration and test data points, $Z_1, \dots, Z_{n+m}$, is exchangeable. Formally, for any permutation $\sigma$ of the indices $\{1, \dots, n+m\}$, the joint distribution of the permuted sequence $(Z_{\sigma(1)}, \dots, Z_{\sigma(n+m)})$ is identical to that of the original sequence $(Z_1, \dots, Z_{n+m})$.
\end{assumption}
Under exchangeability, the test scores are exchangeable with the calibration scores. This allows us to define the conformal uniform variables, which are effectively the normalized ranks of the test samples:
\begin{definition}[Conformal uniform variables]
    For a test sample $Z_{n+j}$, the conformal uniform variable is defined as:
    \begin{equation}
    u_j = \frac{\sum_{i=1}^n \ind\set{s(Z_i) < s(Z_{n+j})} +   U_j}{n+1},
    \label{eqn:intro-conformal-p-values}
    \end{equation}
where $U_1,\ldots,U_m \stackrel{\text{i.i.d.}}{\sim} \text{\em Unif}[0,1]$ are independent of the data. 
\end{definition}
Under Assumption \ref{asmp:no-tie} and \ref{asmp:exchangeability}, each $u_{j}$ is marginally uniformly distributed over $[0,1]$ and this is the same quantity as the null conformal p-value from \cite{gazin2024transductive}. More importantly, we shall see that the vector $\mathbf{u}=(u_{1},...,u_{m})$ admits a tractable joint distribution that does not depend on the underlying data-generating mechanism (Proposition \ref{prop:distribution-oracle-p-values}). We denote this joint distribution by $\mathcal{P}_{n,m}$. This object plays a central role below as it shall serve as the reference distribution for constructing uniform upper bounds on empirical distribution functions.

\subsection[Joint distribution of conformal uniform variables]{Joint distribution of conformal uniform variables}

We revisit a key property of the conformal uniform variables~\eqref{eqn:intro-conformal-p-values}, assuming no ties among the scores. 
Proposition~\ref{prop:distribution-oracle-p-values} below is a well-known result~\citep{gazin2024transductive}, whose proof is omitted for brevity. 
\begin{prop}\label{prop:distribution-oracle-p-values}
    Let \(\{T_i\}_{i=1}^{n+m}\) be independent and identically distributed as
    \(\mathrm{Unif}[0,1]\), and let \(\{U_j\}_{j=1}^{m}\) be independent and
    identically distributed as \(\mathrm{Unif}[0,1]\), independent of
    \(\{T_i\}_{i=1}^{n+m}\). Set 
 \begin{equation}\label{eq:def_qj}
    q_j = \frac{\sum_{i=1}^n \ind\{T_i < T_{n+j}\} + U_j}{n+1}, \quad j = 1, \ldots, m.
    \end{equation}
Then the distribution of $(q_1, \ldots, q_m)$ is identical to the distribution $\cP_{n,m}$ of $(u_1, \ldots, u_m)$ \eqref{eqn:intro-conformal-p-values}.
\end{prop}
This result follows from the exchangeability of the scores: the joint distribution of the conformal uniform variables depends only on the relative ranks of the scores and therefore coincides with that obtained when the scores are replaced by independent $\mathrm{Unif}[0,1]$ random variables. 

Although it is straightforward, Proposition~\ref{prop:distribution-oracle-p-values} enables convenient sampling from the joint distribution of the conformal uniform variables, regardless of the underlying data distribution and score function.



\section[Cornerstone: Upper bound on ecdf of conformal uniform variables]{Main results}\label{sec:upper-bound-oracle}

We introduce a unified strategy for constructing simultaneous upper bounds on the  empirical CDF of conformal uniform variables under our exchangeability assumption. 

Recall the definition of the conformal uniform variables $\{u_j\}_{j=1}^m$ in~\eqref{eqn:intro-conformal-p-values}. The empirical CDF is defined as  
\begin{equation}\label{eqn:ecdf}
    \hat{F}_{n,m}(t) = \frac{1}{m}\sum_{j=1}^m \ind\{u_j \le t\},\quad t\in [0,1].
\end{equation}
The e.c.d.f.~satisfies $\EE[\hat{F}_{n,m}(t)] = t$, while its variance has the form \(\operatorname{Var}\{\widehat F_{n,m}(t)\}=c_{n,m}(t)t(1-t)\); see Proposition~\ref{prop:var-ecdf-conformal-uniform}. 
Given a confidence level $\delta\in (0,1)$, the goal of this section is to find a (possibly random) function $G\colon [0,1]\to [0,1]$ obeying 
\begin{equation}\label{eqn:upper_bound_oracle}
    \PP\left(\hat{F}_{n,m}(t) \le G(t), ~~\forall t \in [0,1]\right) \ge 1-\delta.
\end{equation}
We call such a $G(t)$ an envelope function. 

The problem~\eqref{eqn:upper_bound_oracle} has been investigated in \cite{gazin2024transductive}, who applied the Dvoretzky--Kiefer--Wolfowitz (DKW) inequality~\citep{dvoretzky1956asymptotic,massart1990tight} to obtain a deterministic linear function $G(t)$ of the form
\begin{equation}\label{eqn:baseline}
    G_{\mathrm{DKW}}(t) = t + \lambda_{\mathrm{DKW}},
\end{equation}
where $\lambda_{\mathrm{DKW}}\in \RR$ is a constant depending on $(n,m,\delta)$. We will use their method as a baseline for comparison in our experiments.
An alternative analytic form of $G(t)$ is established in 
\cite{bates2023testing} when the data $\{Z_{i}\}_{i=1}^{n+m}$ are i.i.d.~(stronger than the exchangeability condition here) and  $n\to \infty$. 
In contrast, our idea is to sample from the distribution of conformal uniform variables, which  yields tight and finite-sample valid bounds in various applications.




\subsection{Uniform upper bounds for increasing random functions}


We begin with any one-dimensional summary statistic $T(F)$ such that for any function $F\colon [0,1]\to [0,1]$ and any cutoff $g\in \RR$, there exists an envelope function $\cG(\cdot;g)\colon [0,1]\to[0,1]$ obeying 
\@\label{eq:T_to_G}
T(F)\leq g\quad \Leftrightarrow \quad F(t)\leq \cG (t;g),~\forall t\in [0,1].
\@
The task then reduces to finding a high-probability upper bound $\hat{T}\in \RR$ for $T(\hat{F}_{n,m})$, since 
$$
\PP\big( T(\hat{F}_{n,m})\leq \hat{T}\big)\geq 1-\delta \quad \Leftrightarrow \quad 
\PP\big( \hat{F}_{n,m}(t)\leq \cG(t;\hat{T}) \big)\geq 1-\delta. 
$$
We construct such a cutoff $\hat{T}\in \RR$ by using the $(1-\delta)$-th quantile of multiple independent samples from the joint distribution of the conformal uniform variables.\footnote{For a probability distribution $P$ on $\mathbb{R}$ and $\beta \in (0,1)$, the quantile function is defined as 
$\operatorname{Quantile}(\beta;P)
:=
\inf\{x \in \mathbb{R}: P((-\infty,x]) \ge \beta\}.$}
The procedure is summarized in Algorithm~\ref{alg:monte-carlo}.

\begin{algorithm}
\caption{Sampling for Envelope Function Estimation}
\label{alg:monte-carlo}
\begin{algorithmic}[1] 
\REQUIRE Distribution $\mathcal{D}$ of a random function $F(t), t\in [0,1]$; summary statistic  $T(\cdot): C([0,1]) \rightarrow \mathbb{R}$, envelope function mapping $\cG(\cdot;\cdot)$; confidence parameter $\delta$; number $B$ of Monte Carlo samples.
\STATE Sample $B$ i.i.d.~replications $\{F^{(b)}\}_{b=1}^B$ of $F$ from the distribution $\mathcal{D}$.
\STATE Compute $T^{(b)} = T(F^{(b)})$ for $b = 1, \dots, B$.
\STATE Compute $\hat{T} = \quant (1-\delta ;\sum_{b=1}^B \frac{1}{B+1} \delta_{T^{(b)}} + \frac{1}{B+1}\delta_\infty)$. 
\ENSURE Envelope function $G(\cdot)$ where $G(t) = \cG(t;\hat{T})$ for $t\in [0,1]$.
\end{algorithmic}
\end{algorithm}



Naturally, the shape of our bound  relies on  the summary statistic $T$. 
Motivated by classical empirical process theory, throughout this paper, we mainly consider summary statistics of the general  form 
\begin{equation}\label{eqn:summary-statistics}
    T(F) = \sup_{t \in [\ell, r]} \frac{F(t) - \mu(t)}{\sigma(t)},
\end{equation}
where 
the mean function $\mu(t)$ is typically chosen as the expectation of $F(t)$ (known in the case of $F=\hat{F}_{n,m}$), and $\sigma(t)$ is a template function that controls the shape of the bound. Here, we use two additional parameters $\ell,r\in [0,1]$ to add  flexibility: while it is a natural idea to take the supremum over the entire range $t\in [0,1]$, the fluctuation of $(F(t) - \mu(t))/\sigma(t)$ may be large around $t=0,1$, which leads to conservative bounds on $T(F)$. 
When $[\ell,r]\neq [0,1]$, the region outside  $[\ell,r]$ can be extrapolated using the boundedness and monotonicity of the empirical c.d.f.~$\hat{F}_{n,m}(\cdot)$.

The following theorem formalizes the finite-sample validity of Algorithm \ref{alg:monte-carlo} for a general choice of the summary statistic, and establishes the correspondence in~\eqref{eq:T_to_G}, hence  the form of $G(t)$, based on~\eqref{eqn:summary-statistics}. 

\begin{theorem}\label{thm:monte-carlo-bounds}
Let $F(t), t\in[0,1]$ be a random function drawn from an arbitrary distribution $\cD$. Let $\hat{T}$ be the $(1-\delta)$-quantile of the statistics $T(F)$ estimated via $B\in\NN^+$ Monte Carlo samples as in Algorithm \ref{alg:monte-carlo}. Then
\begin{equation}\label{eqn:bound-monte-carlo}
    1-\delta \le \mathbb{P}(T(F) \leq \hat{T}) \leq 1 - \delta + \frac{1}{B+1},
\end{equation}
where the randomness is over both $F$ and the sampling process. 
Furthermore, if $F(t)$ is increasing in $t$ and bounded above by a constant $M>0$ almost surely, then using the summary statistic defined in \eqref{eqn:summary-statistics}, Algorithm \ref{alg:monte-carlo} produces the piecewise function
    $$
    \cG(t, \hat{T}) = \begin{cases}
        \min\{\mu(\ell) + \hat{T}\sigma(\ell), M\},& \mbox{ if } t <\ell, \\
        \min\{\mu(t) + \hat{T}\sigma(t), M\},& \mbox{ if } \ell \le t \le r, \\
        M, & \mbox{ if } t > r,
    \end{cases}
    $$
which is a valid simultaneous upper bound:
    $$
     1-\delta \le \PP\Big(F(t) \le \cG(t;\hat{T}), ~\forall t\in[0,1]\Big)\le 1-\delta + \frac{1}{B+1}.
    $$

\end{theorem}

Flipping the sign of the summary statistic leads to simultaneous lower bounds, which enables two-sided simultaneous bounds on the e.c.d.f.~when needed. 
We make several remarks on Algorithm \ref{alg:monte-carlo}. 

\vspace{-0.5em}
\paragraph{Sharpness}
As $B\to \infty$, the probability $\PP(T(F)\le G(t,\hat{T}), ~\forall t\in[0,1])$ converges to $1-\delta$. This indicates that our upper bound is asymptotically sharp. In other words, when $B\to \infty$, it is impossible to find another envelope function that is uniformly smaller than $G(\cdot,\hat{T})$ while still maintaining the same level of confidence.

\vspace{-0.5em}
\paragraph{Beyond conformal uniform variables}
While we primarily focus on conformal uniform variables in this work, Algorithm \ref{alg:monte-carlo} can be applied to other multiple testing problems as long as we can sample from the joint distribution of the p-values. We include in  Appendix \ref{app:comparison-iid-conformal} a discussion on extensions to i.i.d.~p-values.

\vspace{-0.5em}
\paragraph{Point-wise bound}
Our strategy  can be easily modified to obtain other types of bounds. For example, if one is interested in an upper bound on $F(t_0)$  at a single point $t_0$, we can set $T(F) = (F(t_0) - \mathbb{E}[F(t_0)])/\sigma(t_0)$ and still follow Algorithm~\ref{alg:monte-carlo} to obtain the cutoff $\hat{T}$. Inverting $T(F)\leq \hat{T}$ then yields the point-wise upper bound $\PP(F(t_0)\leq \sigma(t_0)\cdot \hat{T} + \EE[F(t_0)])\geq 1-\delta$. 

\subsection{Summary statistics}

The general form~\eqref{eqn:summary-statistics} encompasses many well-studied statistics in the literature. Below, we list some representative ones which will be used in our experiments, followed by a numerical comparison. 

\paragraph{One-sided Kolmogorov-Smirnov Statistic} This is the classical statistic defined as
$$
T_\ks(\hat F_{n,m}) = \sqrt{m} \sup_{t\in [0,1]} (\hat F_{n,m}(t) - t). 
$$
When used in Algorithm \ref{alg:monte-carlo}, this statistic yields a linear upper bound function, making it computationally efficient but potentially conservative since it is not adapted to the variance of $\hat{F}_{n,m}(t)$. 

\paragraph{(Truncated) Higher-Criticism Statistics} The Higher-Criticism framework \cite{donoho2004higher} suggests the following statistic: 
\begin{equation}\label{eq:stat_thc}
    T_{\hc, \ell, r, \beta} ( \hat{F}_{n,m}) = \sup_{t\in [\ell,r]} \frac{\hat{F}_{n,m}(t) - t}{(t(1-t))^\beta}. 
\end{equation}


 When $\beta=1/2$, the statistic in \eqref{eq:stat_thc} reduces to
the Higher-Criticism statistic \citep{donoho2004higher}. 
This choice is also motivated by the variance calculation in
Appendix~\ref{app:variance-ecdf}. Proposition~\ref{prop:var-ecdf-conformal-uniform}
shows that
\[
\operatorname{Var}\{\widehat F_{n,m}(t)\}=c_{n,m}(t)t(1-t),
\]
where $c_{n,m}(t)$ accounts for the dependence induced by the shared calibration sample. Hence $t(1-t)$ captures an important component of the
heteroskedasticity of the empirical CDF of conformal uniform variables.
The denominator $(t(1-t))^\beta$ can therefore be viewed as a
variance-adaptive normalization, with $\beta=1/2$ corresponding to
standardization up to the factor $c_{n,m}(t)$. When $n$ is moderately large, the remaining factor
$c_{n,m}(t)$ is comparatively flat over most of the interval; see
Appendix~\ref{app:variance-ecdf}.

As we already briefly mentioned, the truncation parameters $\ell$ and $r$ are crucial. When $\beta  = 1/2$, 
the
supremum over $t\in[0,1]$ is typically attained at very small values of
$t(1-t)$, that is, near the boundaries of the interval \citep{donoho2004higher, li2015higher}, which may lead to overly conservative bounds for mid-range values of $t$. By restricting the range to $[\ell,r]$, we obtain a smaller cutoff $\hat{T}$ and better balanced and practically useful bounds.

The denominator $(t(1-t))^\beta$ controls the shape of the envelope function. 
Additionally, the  parameter $\beta$   determines the extent to which the upper bound is sensitive to the variance of $\hat F_{n,m}(t)$. A smaller $\beta$ results in a bound which exhibits a linear behavior. We include an empirical investigation of the choice of $\beta$ in Appendix \ref{app:beta}.

\paragraph{Berk-Jones Statistics} This statistic is based on a likelihood deviance:
$$
T_\bj = m \cdot \max_{1\le i \le m/2} D(p_{(i)}, i/m),
$$
where $D(p_0, p_1) = p_0 \log(p_0/p_1) + (1-p_0)\log((1-p_0)/(1-p_1))$ is the Kullback-Leibler divergence between Bernoulli distributions with parameters $p_0$ and $p_1$. To understand this, consider the order statistic of the p-values $p_{(1)} \le p_{(2)} \le \ldots \le p_{(m)}$. If we think of the p-values as being i.i.d.~uniform p-values (which is not exactly the case here), we would have $p_{(i)} \approx i/m$, so that $D(p_{(i)}, i/m)$ captures the expected deviation. In our experiments, we find that this statistic leads to bounds that are particularly sensitive to deviations in the middle range of p-values.

\begin{figure}[htbp]
\centering
\begin{subfigure}{0.45\linewidth}
    \centering
    \includegraphics[width=\linewidth]{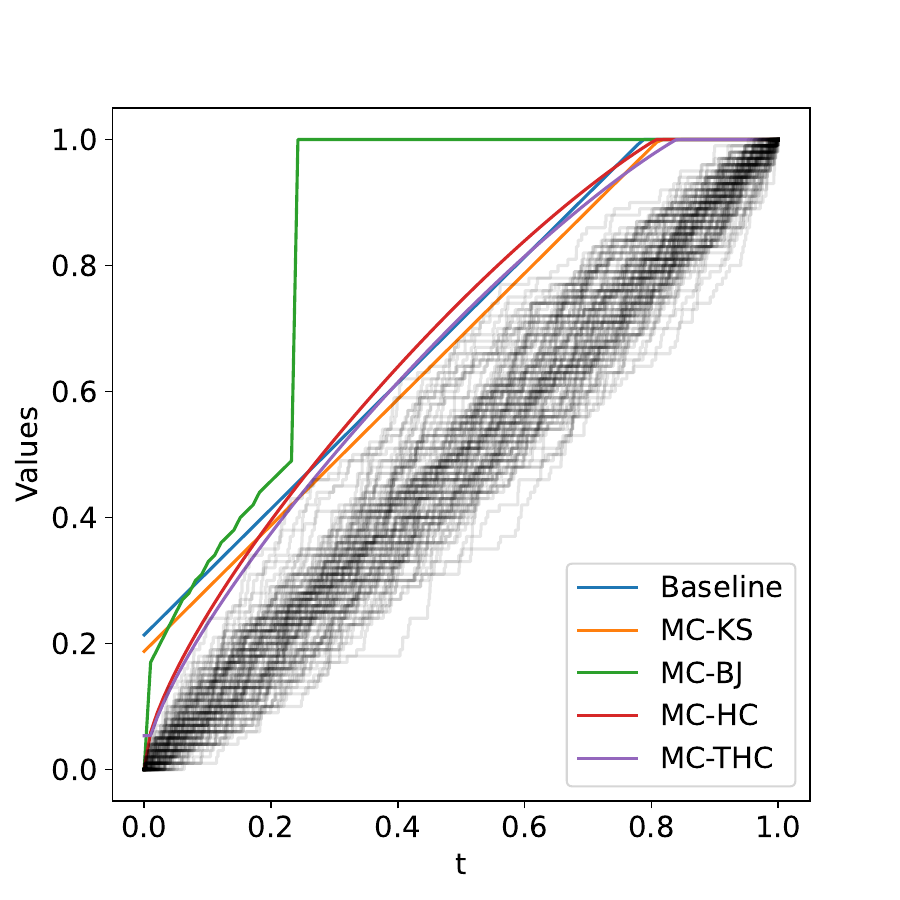}
    \caption{Full range}
\end{subfigure}
\begin{subfigure}{0.45\linewidth}
    \centering
    \includegraphics[width=\linewidth]{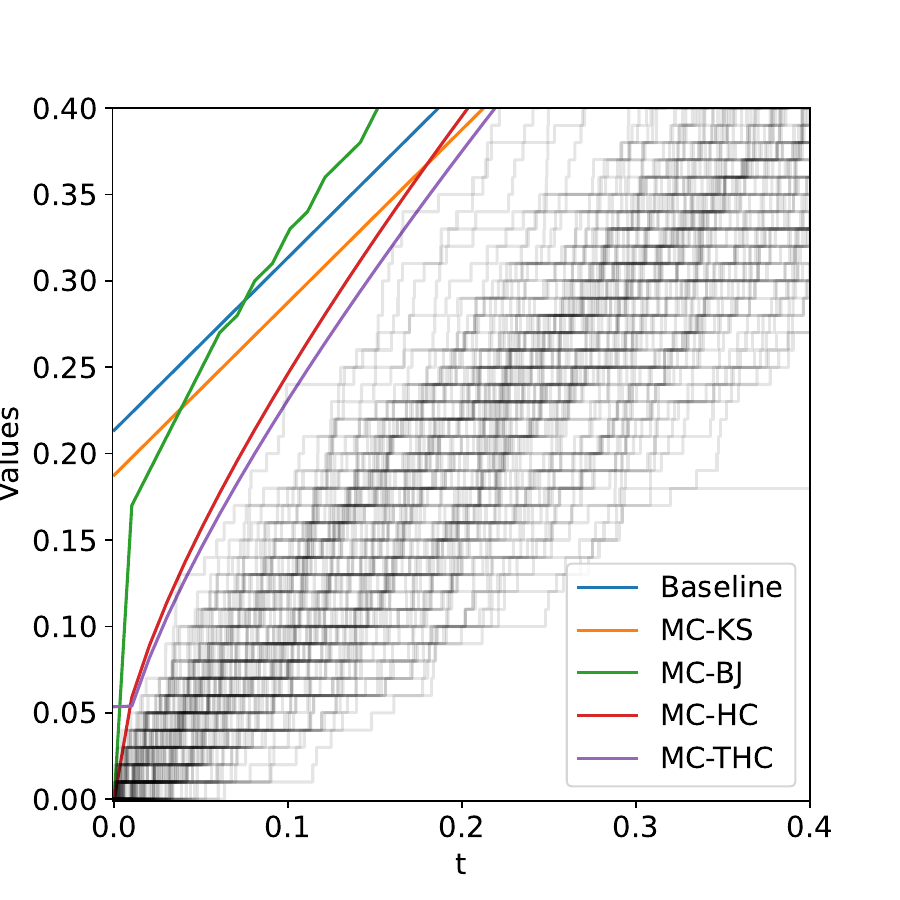}
    \caption{Zoom near $t=0$}
\end{subfigure}
    \caption{Upper bounds on $\hat{F}_{n,m}(t)$ ($n=m=100,$ $\delta=0.1$) constructed via Algorithm~\ref{alg:monte-carlo} with $B=100$. The gray curves represent 100 independent realizations of $\hat{F}_{n,m}$. The colored curves represent different envelope
constructions: MC-KS (Kolmogorov–Smirnov statistic), MC-BJ (Berk–Jones statistic),
MC-HC (Higher-Criticism statistic), and MC-THC (truncated Higher-Criticism statistic).
The Baseline curve corresponds to the linear bound \eqref{eqn:baseline}. 
The MC-THC (purple) and MC-HC bound (red) adapt to the heteroskedasticity $\hat{F}_{n,m}(t)$ (pinching near $t=0$), whereas the Baseline~\eqref{eqn:baseline} and MC-KS bounds are linear and conservative for small values of $t$.
}
    \label{fig:pvals_sim}
\end{figure}


Figure~\ref{fig:pvals_sim} compares our method (based on various summary statistics)   with the approach in \cite{gazin2024transductive} (\texttt{Baseline}). We randomly sample $N=100$ replicates of the p-values for sample sizes $(n,m)=(100,100)$ and plot the corresponding empirical CDFs $\hat F_{n,m}(t)$ as gray curves. 
The gray curves show that the fluctuation of $\hat{F}_{n,m}(t)$ is smaller when $t$ is near the boundaries. The \texttt{Baseline} which uses a linear envelope function fails to capture such  heteroskedasticity and leads to conservative bounds for  $t\in[0,0.2]$, which could be the region of most interest in many applications. Our method with the KS statistic also yields a linear envelope function; the data-driven cutoff is visibly tighter than the DKW-type cutoff, though it also fails to adapt to the heteroskedasticity. 
By contrast, our method with the Higher-Criticism statistic with $(\ell, r,\beta)=(0,1,1/2)$ (\texttt{MC-HC}) leads to substantially less conservative bounds, especially for small values of $t$.
As we briefly mentioned before, the choice of $\sigma(t)$ in the Higher Criticism statistic can yield large fluctuations of $(F(t)-\mu(t))/\sigma(t)$ around $t\approx 0$, which in turn translates into increased conservativeness for $t\approx 0.5$. To address this, the truncated HC statistic (\texttt{MC-THC}) only takes the supremum over $t\in [\ell,r]=[0.01,0.99]$ and extrapolates the bounds in regions outside  $[\ell,r]$. The truncated HC statistic yields envelopes that remain close to the empirical CDF across most of the domain $t\in[0,1]$, with the exception of a narrow region around $t\approx 0.5$, where the KS statistic yields a slightly closer envelope.

The Berk-Jones statistic (\texttt{MC-BJ}), despite its theoretical appeal, tends to be overly conservative across most of the domain except for very small values of $t$ (below 0.1), where it occasionally provides competitive bounds.



\subsection{Computational efficiency}

Computing the Kolmogorov-Smirnov and truncated Higher-Criticism statistics is particularly efficient because of monotonicity: Fact~\ref{fact:monotone} is proved in Appendix~\ref{app:upper-bound-oracle}. 

\begin{fact}\label{fact:monotone}
For any $x\in[0,1]$ and $\beta \in [0,1]$,   $f(t) = \frac{x-t}{(t(1-t))^\beta}$ is monotonically decreasing in $t\in [0,1]$.
\end{fact}

This fact allows us to characterize the computation efficiency of these summary statistics and ensures that our bounds remain scalable even when $m$ is very large.

\begin{prop}[Computational cost of summary statistics]\label{prop:comp-cost}
\quad
\begin{enumerate}
    \item \textbf{Truncated Higher-Criticism (THC).} 
    Computing the truncated higher-criticism statistic~\eqref{eq:stat_thc} for any value of $(\ell,r,\beta)$ 
    reduces to checking the $m$ jump points of $\hat F_{n,m}$ within $[\ell,r]$. Hence the statistic can be computed in $\mathcal{O}(m)$ time.

    \item \textbf{Evaluating the envelope at a fixed $t$.} 
    Given the cutoff $\hat T$ from Algorithm~\ref{alg:monte-carlo}, the upper bound function
    $ 
    G(t) = \min\{\mu(t) + \hat T \cdot \sigma(t),\, 1\}, \quad t \in [\ell,r],
    $ 
    can be evaluated at any fixed value of $t$ in constant time $\mathcal{O}(1)$. For $t\notin [\ell,r]$, the extension rule in Theorem~\ref{thm:monte-carlo-bounds} also takes $\mathcal{O}(1)$ time.
\end{enumerate}
\end{prop}

\subsection{Simultaneous FCP upper bounds}\label{sec:fcp_bounds}
An immediate  application of our upper bound on $\hat{F}_{n,m}$ is bounding the false coverage proportion (FCP) for conformal prediction sets, following ideas similar as in~\citep{gazin2024transductive}. 

In this context, each data point is $Z_i = (X_{i},Y_i)$ for $i \in [n+m]$. For a given test instance with only $X_{n+j}$ observed and a confidence level $\alpha\in (0,1)$, we construct a prediction set  $\cC(X_{n+j};\alpha) = \set{y: u_j(X_{n+j},y)> \alpha}$, where $u_j(X_{n+j},y)$ is obtained by replacing $Z_{n+j}$ with $(X_{n+j},y)$ in~\eqref{eqn:intro-conformal-p-values}. 
With exchangeable data, the conformal prediction set is guaranteed to cover the unknown outcome with  $\PP(Y_{n+j} \in \cC(X_{n+j};\alpha) )$ for each $j\in[m]$; see~\cite{vovk2005algorithmic}.
The FCP is the proportion of prediction sets that fail to cover the true labels among all $m$ prediction sets, and is equal to (noticing that $u_j = u_j(X_{n+j},Y_{n+j})$)
\begin{equation}
    \fcp(\alpha) = \frac{1}{m}\sum_{j=1}^m \ind\{Y_{n+j} \not\in \cC(X_{n+j};\alpha)\} = \frac{1}{m} \sum_{j=1}^m \ind\{u_j \le \alpha\} = \hat{F}_{n,m}(\alpha).
\end{equation}
Thus, the FCP at a confidence level $\alpha\in (0,1)$ is  directly related to the empirical CDF of the conformal uniform variables at the value $\alpha$. 
Therefore, Figure \ref{fig:pvals_sim} can also be viewed as illustrating the upper bound on the FCP.  

\section{Simultaneous FDP bounds in outlier detection}
\label{sec:outlier-detection}

In this section, we establish simultaneous FDP upper bounds for the outlier detection problem introduced in \citep{bates2023testing}. 
Recall that we observe a calibration dataset $\cD_\calib = \set{X_i}_{i\in[n]}$ where the data points are i.i.d. samples from an unknown distribution $\cP_0$, and a test dataset $\cD_\test = \set{X_{n+j}}_{j\in[m]}$ where each test sample is independently sampled from an unknown distribution. (Outliers may be sampled from unequal and arbitrary distributions.) 
The goal of outlier detection is to identify test samples whose distributions deviate from $\cP_0$.

\subsection{Setup}
The outlier detection problem can be formulated as a multiple testing problem. 
For each observation $X_{n+j}$, we test the null hypothesis $H_{n+j}: X_{n+j} \sim \cP_0$. In other words, under the null hypothesis, the test sample  follows the same distribution as the calibration data.
In many standard solutions, a one-class classifier is trained on a separated dataset to obtain a score function $s\colon \cX\to \RR$ \citep{pimentel2014review}, where higher values of $s(X)$ indicate stronger evidence that $X$ follows the same distribution as the calibration data. As before, we assume $s(X_1)$ is a continuous random variable.

\cite{bates2023testing} shows that conformal p-values wrapped around any score function can be leveraged to test for outliers with FDR control at a specified level. 
The conformal p-value takes exactly the form of~\eqref{eqn:intro-conformal-p-values}:
\begin{align*}
    p_j = \frac{\sum_{i=1}^n \ind\{s(X_i)< s(X_{n+j})\} +  U_j}{n+1}, \quad j \in [m].
\end{align*}
A  small value of $p_j$ means that $s(X_{n+j})$ is comparably small and thus suggests  that $X_{n+j}$ is likely to be an outlier. 
The joint distribution of these p-values can be characterized under mild assumptions. 

\begin{assumption}\label{assumption:exchangeable-outlier-detection}
    Let $\mathcal{H}_0 \subseteq [m]$ denote the index set of true inliers, where $j \in \mathcal{H}_0$ if and only if $X_{n+j}$ follows the same distribution $\cP_0$ as the calibration data. We assume that the scores of the calibration samples and inlier test samples $\{s(X_i)\}_{i \in [n] \cup \{n+j: j\in\mathcal{H}_0\}}$ are exchangeable random variables. 
\end{assumption}

Beyond pre-trained scores, Assumption~\ref{assumption:exchangeable-outlier-detection} permits the use of so-called \emph{adaptive} score functions learned from data~\cite{gazin2024transductive,marandon2024adaptive}. Such adaptive scores are no longer i.i.d., yet they are still exchangeable and obey this assumption.
To avoid unnecessary complication, we again assume the scores have no ties. Proposition~\ref{prop:distribution_pval_outlier} provides a unified description of the joint distribution of the inlier conformal p-values.

\begin{prop}\label{prop:distribution_pval_outlier}
Under Assumptions \ref{asmp:no-tie} and \ref{assumption:exchangeable-outlier-detection}, letting $m_0 = |\cH_0|$ denote the number of true inliers in test samples, the conformal p-values $\{p_j\}_{j \in \cH_0}$ follow the  joint distribution $\cP_{n,m_0}$ given in Proposition \ref{prop:distribution-oracle-p-values}.
\end{prop}



For a rejection threshold $t\in[0,1]$, the FDP is given by
\@
\label{eq:def_FDP_outlier}
\mbox{FDP}(t) = \frac{\sum_{j\in \mathcal{H}_0} \ind\{j \in \cR(t)\}}{\max\{1,|\cR(t)|\} } = \frac{\sum_{j\in \mathcal{H}_0} \ind\{p_j \le t\}}{\max\{1, \sum_{j=1}^m \ind\{p_j \le t\}\} }.
\@
Since we observe the denominator in~\eqref{eq:def_FDP_outlier}, the key to establishing a simultaneously valid upper bound for $\mbox{FDP}(t)$ is to bound the unobserved number of false discoveries $\sum_{j\in \mathcal{H}_0} \ind\{p_j \le t\}$.


\subsection{Bounding the number of false selections}

The following lemma connects the distribution of $\{p_j\}_{j\in[m]}$ to the distribution $\cP_{n,m}$. Without loss of generality, we assume $\cH_0 = \{1,\dots,m_0\}$. 

\begin{lemma}[Coupling with completed conformal random variables]\label{lem:coupling}
Under Assumptions~\ref{asmp:no-tie} and~\ref{assumption:exchangeable-outlier-detection}, there exists a collection of random variables
$ 
(p_1^\ast,\dots,p_m^\ast)
$ 
such that:
\begin{enumerate}
    \item \emph{Null preservation}: 
    $ 
    (p_1^\ast,\dots,p_{m_0}^\ast) = (p_1,\dots,p_{m_0})$ almost surely.
    
    \item \emph{Joint distribution}: The full vector $(p_1^\ast,\dots,p_m^\ast)$ has joint distribution
    $\cP_{n,m}$ as defined in Proposition~\ref{prop:distribution-oracle-p-values}.
\end{enumerate}
\end{lemma}

Lemma~\ref{lem:coupling} states that we can always couple the null p-values with certain $p_j^*$'s  for those $j\notin \cH_0$, such that the entire sequence of p-values follows the known, tractable distribution $\cP_{n,m}$. The key to establishing simultaneous upper bounds for the FDP is that a valid upper bound for the new p-value sequence leads to a valid upper bound for the original p-value sequence with unknown non-null p-value distributions. Namely, 
\begin{equation}\label{eqn:inequality-p-p*}
    \sum_{j\in \mathcal{H}_0} \ind\{p_j \le t\} = \sum_{j\in \mathcal{H}_0} \ind\set{p_j^*\le t} \le \sum_{j=1}^m \ind\set{p_j^* \le t}.
\end{equation}
The next theorem shows that any simultaneous upper bound for $\sum_{j=1}^m \ind\set{p_j^* \le t}$
immediately yields a simultaneous FDP bound for the observed p-values.

\begin{theorem}[Simultaneous FDP bound via completed sequences]\label{thm:fdp-outlier-detection-naive}
Under Assumptions~\ref{asmp:no-tie} and~\ref{assumption:exchangeable-outlier-detection}, let $(p_1^\ast,\dots,p_m^\ast)$ be constructed as in
Lemma~\ref{lem:coupling}.
Suppose $G:[0,1]\to[0,1]$ is a  function independent of $\{p_j^*\}_{j=1}^m$ and $\{p_j\}_{j=1}^m$ which satisfies
\begin{equation}\label{eq:completed-envelope}
\mathbb{P}\Bigg(
\frac{1}{m}\sum_{j=1}^m \mathbf{1}\{p_j^\ast \le t\} \le G(t),
\ \forall t\in[0,1]
\Bigg) \ge 1-\delta.
\end{equation}
Then if we work with the rejection sets $R(t)=\{j:p_j\le t\}$, we have  
\[
\mathbb{P}\Bigg(
\mathrm{FDP}(t)
\le
\frac{m\,G(t)}{1 \vee \sum_{j=1}^m \mathbf{1}\{p_j \le t\}},
\ \forall t\in[0,1]
\Bigg) \ge 1-\delta.
\]
\end{theorem}

The envelope function $G(t)$ in Theorem~\ref{thm:fdp-outlier-detection-naive} can be constructed using
the procedure from Section~\ref{sec:upper-bound-oracle}, by sampling from the known joint distribution
$\cP_{n,m}$.
Importantly, $G(t)$ depends only on $(n,m,\delta)$ and on the choice of summary statistic, and not on
the particular realization of the p-values.

The above FDP upper bound  can be conservative because the upper bound $\sum_{j=1}^m \ind\set{p_j^* \le t}$ on the number of false discoveries is rather crude. 
Next, we explore two additional techniques to refine this initial impulse.

\subsection{Enhancing the FDP upper bound}
The main new ingredient in our outlier-detection bound is the envelope construction from Section~3. 
The refinements below are modular: the self-refinement step is adapted from \citep{hemerik2019permutation}, while the null-proportion tightening is related to the strategy of \citep{gazin2024transductive}. These refinements can be applied on top of any valid simultaneous upper bound on the null empirical process. 
\subsubsection{A tighter bound on false discoveries}

Our first strategy is an improved upper bound on the number of false discoveries, following~\citep{hemerik2019permutation}. 
Given $\delta \in (0,1)$, suppose one has access to a function $B(t)$ obeying
\begin{equation}
    \label{eqn:goal}
    \PP\big(|\cR(t) \cap \cH_0| \le B(t), \forall t \in (0,1)\big) \ge 1-\delta.
\end{equation}
We consider a ``self-improved" version
\begin{equation}\label{eqn:improved_Bt}
    B^\star(t) = |\cR(t)| - \sup_{s\le t, s \in (0,1)}\{|\cR(s)| - B(s)\}_+ \le B (t).
\end{equation}
The following result shows that ${B^\star(t)}/{(1\vee |\cR(t)|})$ is also a valid upper bound on the FDP.
\begin{prop}[Self-refinement]\label{prop:tighten}
    If the number of false discoveries obeys $\PP(|\cR(t) \cap \cH_0)| \le B(t), \forall t \in (0,1)) \ge 1-\delta$, then
    \[
    \PP\bigg(\fdp(\cR(t)) \le \frac{B^\star(t)}{1\vee |\cR(t)|}, \forall t \in (0,1)\bigg) \ge 1-\delta,
    \]
    where $B^\star(t)$ is defined in \eqref{eqn:improved_Bt}.
    Moreover, $B^\star(t)$ can be computed efficiently as
    \begin{align}
        B^\star(t) = \min_{j: p_j\le t} \{B(p_j) + |\cR(p_{\le t})| - |\cR(p_j)|\},
    \end{align}
    where $p_{\le t}$ denotes the largest p-value not exceeding $t$.
\end{prop}
At a high level, the proposition leverages the fact that when enlarging the selection set, the increase in the number of false discoveries cannot exceed the increase in the number of discoveries. Throughout this section and the next, we will apply this refined bound $B^\star(t)$ when constructing FDP upper bounds.

\subsubsection{Upper bound on the number of inliers}

Our second strategy is based on upper bounding the number of inliers, that is, the number of true nulls, $m_0$. The bound in Theorem \ref{thm:fdp-outlier-detection-naive} may be conservative because we substitute the summation over $\cH_0$ with that over the entire set of hypotheses in \eqref{eqn:inequality-p-p*}, and the number of nulls $|\cH_0| = m_0$ with $m$. 
Should $m_0$ be known, one can construct the envelope function $G(t)$ based on $\{p_j^*\}_{j\in \cH_0}\sim \cP_{n,m_0}$ instead of $\cP_{n,m}$.
This observation suggests developing a more informative bound on $m_0$. 

We will employ a technique similar to that of \cite{gazin2024transductive} to establish such an upper bound. Consider the event 
$$
\cE = \Big\{ \textstyle{\sum_{j\in\cH_0}}\ind\set{p_j^* \le t} \le G_{m_0}(t), \forall ~ t\in[0,1] \Big\},
$$
where $G_{m_0}(\cdot)$ is a function which depends on the number of true nulls $m_0$.
On the event $\cE$, it holds that 
$$
\textstyle{\sum_{j=1}^{m}} \ind\set{p_j > t} \ge \textstyle{\sum_{j\in\cH_0}} \ind\set{p_j^* > t} \ge m_0 - G_{m_0}(t), \forall ~ t\in[0,1],
$$
which further implies 
$$
m_0 \le \max\set{r: \textstyle{\sum_{j=1}^{m}} \ind\set{p_j > t} \ge r - G_r(t),  \forall ~t\in[0,1]} =: \hat{m}_0.
$$

\subsubsection{Integrating the two strategies}

With the techniques discussed above, we can now construct a sharper simultaneous upper bound on the FDP. 
\begin{theorem}\label{thm:fdp-upper-bound-outlier-detection}
Suppose $G_k(\cdot)$ obeys
    $\PP(\sum_{j=1}^k \ind\set{p_j^*\le t} \le G_k(t), \forall t \in [0,1]) \ge 1-\delta $ for each $k \in \{1, \dots, m\}$. 
Define  
$$
\hat{m}_0 := \max \Big\{ r \in \{0,1,\dots,m\} \colon  \textstyle{\sum_{j=1}^m} \ind\{ p_j > t \} \ge r - G_r(t), \forall t \in [0,1] \Big\}.
$$
Let $B(t) := \sup_{k \le \hat{m}_0} G_k(t)$. Then with probability at least $1-\delta$, it holds that
$$
\textnormal{FDP}(t) \le \frac{ B^{}(t) \wedge |\mathcal{R}(t)| }{ 1 \vee |\mathcal{R}(t)| }, \quad \forall ~ t \in [0,1].
$$
Furthermore, the bound holds when replacing $B(t)$ with  $B^\star(t) = \min_{s \le t} \{ B(s) + |\mathcal{R}(t)| - |\mathcal{R}(s)| \}$.
\end{theorem}

The computational cost of the proposed procedure is modest. Once the summary
statistics have been computed, evaluating the envelope function $G_r(t)$ at a
fixed value of $t$ requires only $O(1)$ operations. Note that the envelope function $G_r(t)$ can be taken to be nondecreasing in $t$
(by replacing it with its monotone envelope if necessary). This property
simplifies the verification step below.

To compute the estimator
$\hat m_0$, we must verify the condition
\[
\sum_{j=1}^m \mathbf{1}\{p_j > t\} \ge r - G_r(t), \qquad \forall t \in [0,1].
\]
Since $\sum_{j=1}^m \mathbf{1}\{p_j > t\}$ is a step function with at most $m$
discontinuity points and $G_r(t)$ is nondecreasing, this verification reduces to checking these finitely many
points. For a fixed value of $r$,
this step requires $O(m)$ time. As a result, computing $\hat m_0$ requires
$O(m^2)$ operations overall.

Once the functions $\{G_k(t)\}_{k\le \hat m_0}$ are available, evaluating
$B(t)=\sup_{k\le \hat m_0} G_k(t)$ at a fixed value of $t$ costs $O(m)$ time. By Proposition~\ref{prop:tighten},
$B^\star(t)$ only needs to be evaluated at $m$ candidate p-values at most.
Consequently, computing $B^\star(t)$ requires $O(m^2)$ time in total.
Therefore, both $\hat m_0$ and $B^\star(t)$ can be computed in quadratic time
with respect to $m$.

\subsection{Experiments}
We now evaluate the empirical performance of our FDP upper bound on synthetic data. Following the setup of \cite{gazin2024transductive}, each data point $X_i \in \mathbb{R}^{50}$ is drawn from a Gaussian mixture distribution $\cP_0^a$ of the form
\[
X_i = \sqrt{1+a}\,V_i + W_i,
\]
where $V_i \sim N(0,I_{50})$ has independent standard Gaussian components, and $W_i$ is drawn uniformly from a discrete support set $W \subseteq [-3,3]^{50}$ of cardinality $|W|=50$. The vectors in $W$ are sampled once at the start of the experiment. The parameter $a \geq 0$ controls the signal strength:  the inliers follow the distribution $\cP_0^0$, and  a larger value of $a$ leads to outliers that are more separated from the inliers. 
We repeat each experiment $K=100$ times. In each replication, we train a one-class SVM on $n_{\mathrm{train}}=1000$ samples from $\cP_0^0$, use $n_{\mathrm{calib}}=1000$ inlier samples to construct conformal p-values, and generate $n_{\mathrm{test}}=1000$ test samples, of which $90\%$ are inliers from $\cP_0^0$ and $10\%$ are outliers from $\cP_0^a$.

\paragraph{Validity of FDP bounds}
Figure~\ref{fig:exp-fdp-diff} plots the difference between our bound (constructed with the truncated higher-criticism statistic, $[\ell,r]=[0.01,0.99]$, $\beta=1/2$) and the true FDP across thresholds $t$. The empirical coverage across 100 replications is $0.96$, exceeding the target $1-\delta=0.9$. The gap between the bound and the true FDP narrows as $t$ increases, reflecting the fact that larger rejection sets stabilize the denominator of the FDP and reduce variance.
\begin{figure}[t]
    \centering
    \includegraphics[width=0.8\linewidth]{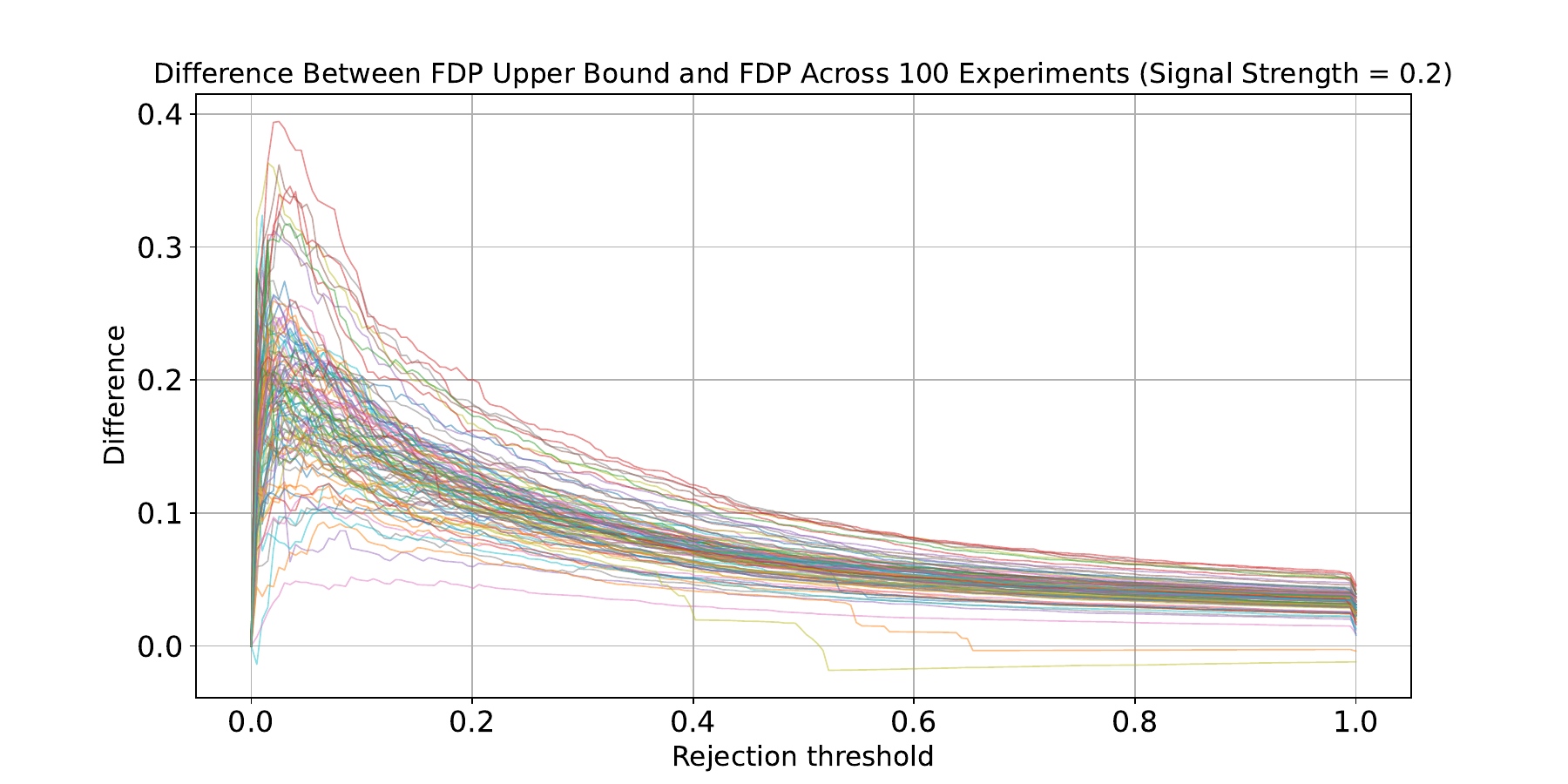}
    \caption{\textbf{Empirical coverage of the FDP bound.}
The plot displays the difference between our FDP upper bound (using MC-THC) and the true FDP across 100 replications of the outlier detection task ($n=m=1000$, signal strength $a=0.2$, target $1-\delta=0.9$). The curves remain above zero in $96\%$ of the trials, demonstrating validity. The variance of the bound decreases as the rejection threshold $t$ increases.
    }
    \label{fig:exp-fdp-diff}
\end{figure}

\begin{figure}[htbp]
    \centering
      \centering
      \includegraphics[width=\linewidth]{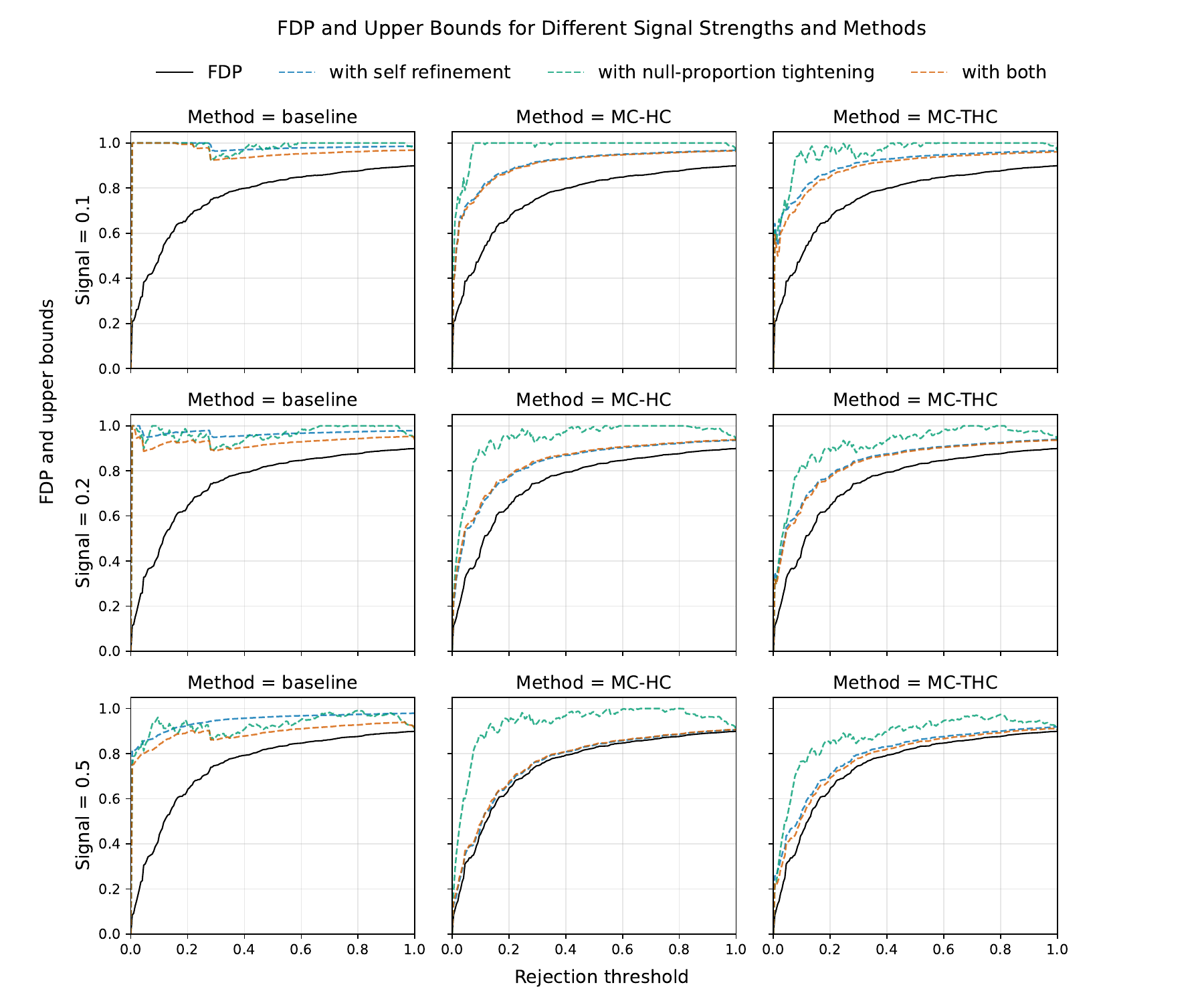}
      \caption{\textbf{Impact of refinement strategies across signal strengths.}
FDP upper bounds in outlier detection with fixed purity ($90\%$) and varying signal strength ($a \in \{0.1, 0.2, 0.5\}$). 
Columns correspond to different FDP envelopes. The curves compare three refinement levels: bounding the null count $\hat{m}_0$, the self-refinement step from Proposition \ref{prop:tighten}, and the combined strategy. The combined approach consistently yields the sharpest envelope.}
      \label{fig:outlier_fixed_purity}
\end{figure}

\begin{figure}[htbp]
    \centering
      \centering
      \includegraphics[width=\linewidth]{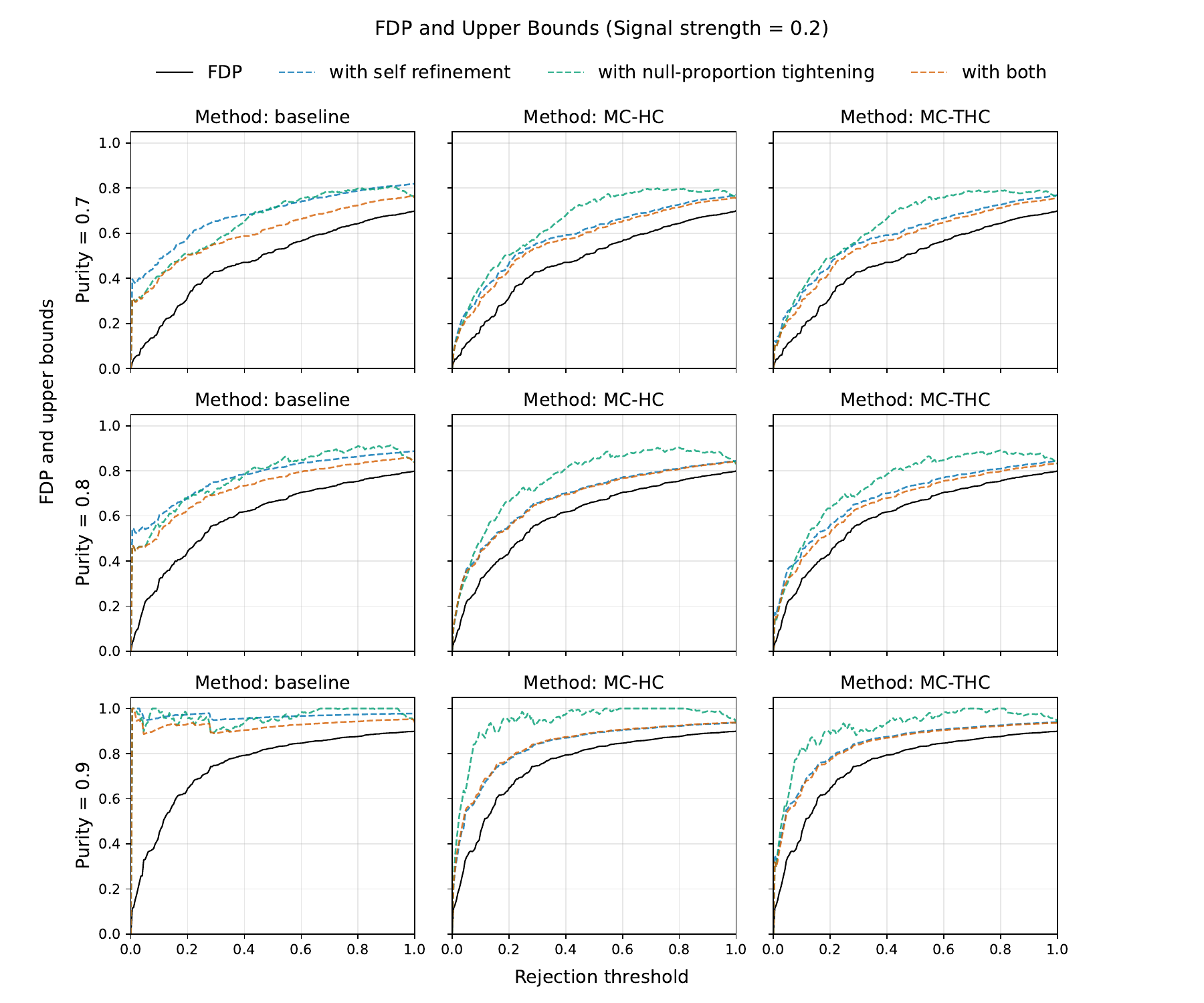}
      \caption{\textbf{Impact of refinement strategies across purity levels.}
FDP upper bounds with fixed signal strength ($a=0.2$) and varying inlier purity ($\{70\%, 80\%, 90\%\}$). As the proportion of outliers increases (lower purity), the benefit of the $\hat{m}_0$ tightening becomes more pronounced compared to the self-refinement step alone. }
      \label{fig:outlier_fixed_strength}
\end{figure}

\paragraph{Comparison of tightening strategies}

We then evaluate the effectiveness of the two strategies introduced in Theorem~\ref{thm:fdp-upper-bound-outlier-detection}: the self-improved version from Proposition~\ref{prop:tighten} and the sharper upper bound on the number of nulls $\hat m_0$. Figure~\ref{fig:outlier_fixed_purity} reports the FDP bounds under fixed purity (fraction of inliers in the test data) of $0.9$ across different signal strengths. Each panel shows three curves: using only $\hat m_0$, using only Proposition~\ref{prop:tighten}, and using both. Figure~\ref{fig:outlier_fixed_strength} presents the results with the same methods under a  signal strength at $a=0.2$ while varying the purity level.  

The results highlight several key patterns. First, as the signal strength increases, our proposed FDP upper bound becomes increasingly sharp, as the detection becomes easier with improved separability between inliers and outliers. Second, the null-proportion tightening strategy yields a larger efficiency gain than the self-improvement strategy alone, yet using both consistently gives the sharpest bound. The relative contribution of the two refinement strategies varies visibly with the purity level. In Figure~\ref{fig:outlier_fixed_strength}, when the purity is high, the self-improvement step dominates, while at lower purity levels, the additional benefit of tightening the upper bound on $\hat m_0$ becomes more visible. Finally, our combined strategy yields substantially tighter FDP bounds than the baseline upper bound of \citep{gazin2024transductive} across all regimes. As illustrated in Figure~\ref{fig:pvals_sim}, the linear envelope function fails to adapt to the heteroskedasticity of the empirical CDF, and is conservative particularly near the boundaries ($t \approx 0$). In contrast, our method leverages the variance-adaptive MC-THC statistic with additional tightening strategies to produce a sharper upper bound. 

\FloatBarrier
\section{Controlling FDP/precision in conformal selection}\label{sec:conformal-selection}

In this section, we demonstrate how our method can be used to obtain simultaneous FDP bounds for the selection sets generated by conformal selection \citep{jin2023aselection,jin2023bmodel}. 
In this problem, one has access to a calibration dataset $\cD_\calib = \{(X_i,Y_i)\}_{i\in[n]}$, where $X_i\in \cX$ are the features, and $Y_i\in \RR$ is the outcome of interest. 
We also have a test dataset $\cD_\test = \{X_{n+j}\}_{j\in [m]}$ with unknown outcomes $\{Y_{n+j}\}_{j\in[m]}$ and wish to identify test samples where the unknown outcomes are substantially large with $Y_{n+j}>c_{n+j}$; $c_{n+j}\in \RR$ are pre-specified and potentially random thresholds. 

\subsection{Setup}


The above problem can be viewed as a binary classification problem over the test samples. The positive class consists of samples whose outcomes exceed their corresponding thresholds, that is, samples with $Y_{n+j} > c_{n+j}$. A selection set $\mathcal R \subseteq [m]$ is therefore the set of test samples predicted to be positive. In this language, the precision of $\mathcal R$ is the fraction of selected samples that are truly positive,
\[
\mathrm{Precision}(\mathcal R)
=
\frac{\sum_{j=1}^m \ind\{j\in\mathcal R,\,Y_{n+j}>c_{n+j}\}}{1\vee |\mathcal R|}.
\]
Equivalently, the false discovery proportion is one minus precision:
\[
\mathrm{FDP}(\mathcal R)
=
\frac{\sum_{j=1}^m \ind\{j\in\mathcal R,\,Y_{n+j}\le c_{n+j}\}}{1\vee |\mathcal R|}.
\]
Thus, an upper bound on the FDP gives a lower bound on the precision of the selected set.



The conformal selection procedure of \cite{jin2023aselection} identifies a selection set $\cR$ such that the FDR, $\EE[\fdp(\cR)]$, is below a pre-specified level $\alpha \in (0,1)$. This method uses a calibration dataset $\{(X_i, Y_i, c_i)\}_{i \in [n]}$. We assume the following threshold observation and exchangeability condition.
\begin{assumption}
    The random triplets $\{(X_{i}, Y_{i},c_i)\}_{i \in [m+n]}$ are exchangeable across $i \in [n+m]$.  This condition is automatically satisfied when $\{(X_i,Y_i)\}_{i=1}^{n+m}$ are exchangeable and either $c_i=c$ for every $i\in[n+m]$ for some fixed constant $c$, or $c_i=g(X_i)$ for every $i\in[n+m]$ for some fixed function $g$.
\end{assumption}

To assess the confidence in rejecting $H_j$, conformal selection constructs conformal p-values based on any nonconformity score $V$ that satisfies the following monotone property.

\begin{definition}
    A nonconformity score $V(\cdot, \cdot,\cdot): \mathcal{X} \times \mathcal{Y}\times \mathbb{R} \mapsto \mathbb{R}$ is monotone if $V(x,y,c) \le V(x, y',c)$ for any $x\in \mathcal{X}, c\in\mathbb{R}$, and any $y,y'\in \mathcal{Y}$ with $y\le y'$.
\end{definition}
Generally, the nonconformity score should measure how extreme the value $y$ is compared to the typical behavior of the outcome for a given feature value and threshold $c$. For instance, one could define $V(x,y,c) = y - \hat\mu(x)$, where $\hat\mu(\cdot)$ is a prediction algorithm estimating the conditional mean of $Y$ given $X$.
 
Conformal selection constructs the p-values  
$$
p_j^\cs = \frac{\sum_{i=1}^n \ind \{V(X_i,Y_i, c_i) <  V(X_{n+j}, c_{n+j}, c_{n+j})\} +  U_j}{n+1},
$$
where $U_j \sim \mbox{Unif}(0,1)$ are i.i.d. random variables.  

As it is natural to select p-values of the smallest magnitudes, we  focus on selection sets of the form $\cR(t) = \set{j\in[m]:p_j^\cs \le t}$. The conformal selection set in~\cite{jin2023aselection} coincides with $\cR(t)$ when  $t\in \RR$ is the cutoff determined by running the BH procedure on $\{p_j^\cs\}$ at nominal level $\alpha$. 
In practice, a uniformly valid upper bound on $\fdp(\cR(t))$ across $t\in [0,1]$ provides useful information on the quality of a series of selection sets, allowing practitioners to flexibly balance power and selection error control.

To prepare, we define the ``oracle'' conformal uniform variables 
\@\label{eq:oracle_pval_cs}
p_j^{\cs,\oracle} = \frac{\sum_{i=1}^n \ind \{V(X_i,Y_i,c_i) < V(X_{n+j},Y_{n+j},c_{n+j})\} +  U_j}{n+1}.
\@
Here, the uniform random variables $U_j$'s are the same as those in $p_j^{\cs}$'s. It is straightforward to verify that the $p_j^\oracle$'s are exchangeable random variables and follow the distribution $\cP_{n,m}$ from Proposition \ref{prop:distribution-oracle-p-values}.
While they are unobservable due to the unknown labels $Y_{n+j}$, these oracle p-values are key to constructing our upper bounds.  

\subsection{Methodology}

A key result in conformal selection is that we can bound the number of false discoveries using ``oracle" conformal p-values $p_j^\oracle$, as shown in the following lemma.

\begin{lemma}
\label{lem:conformal-selection-fdp-inequality}
Consider the set of null hypotheses $\mathcal{H}_0 = \{j \in [m]: Y_{n+j} \le c_{n+j}\}$. For any monotone score function $V$ and the selection set $\cR(t) = \{j \in [m]: p_j^\cs \le t\}$, it holds for every $t \in [0,1]$ that 
\begin{equation}
\label{eq:upper_bound_cs_fd}
|\cR(t) \cap \cH_0| \le \textstyle{\sum_{j=1}^m} \ind \{p_j^{\cs,\oracle} \le t\}.
\end{equation}
Furthermore, the above inequality becomes an equality for $t < \frac{1}{n+1}\sum_{i=1}^n \ind\{Y_i \le c_i\}$ if for any $x\in \cX$, one has (1) $V(x,y,c) > V(x,c,c)$ for any $y > c$; and 
(2)   $V(x,y,c) = V(x, c, c)$ for  $y \le c$. 
\end{lemma}

Lemma \ref{lem:conformal-selection-fdp-inequality} shows that it suffices to construct a valid upper bound on the empirical CDF of $\{p_j^{\cs,\oracle}\}_{j=1}^m$, which can be done even though these oracle p-values are not observed, given that their distribution is known. 
The second part of Lemma~\ref{lem:conformal-selection-fdp-inequality} hints at the score choice that leads to sharp upper bounds.  
An example that satisfies the equality conditions in Lemma~\ref{lem:conformal-selection-fdp-inequality} is   the clipped score function 
$$
V(x,y,c) = M\ind\{y>c\} + c\ind\{y\leq c\} - \hat\mu(x),
$$ 
where $M > 2\sup_x |\hat\mu(x)|$ and $\hat{\mu}(x)$ is any predictor of $\EE[Y\given X=x]$. Such clipped scores are shown to be powerful in~\cite{jin2023aselection} and a standard choice in conformal selection methods~\cite{bai2024optimized,nair2025diversifying,gui2025acs}. In general, given any monotone score function $V'(x,y,c)$ specified by the user, one can always construct 
\$
V(x,y,c) = M\ind\{y>c\} + V'(x,y,c) \ind\{y\leq c\},
\$
where $M>0$ is a sufficiently large constant, or effectively, $M=+\infty$, that preserves the relative ranking of p-values based on $V'$ while achieving higher power and satisfying the conditions in Lemma~\ref{lem:conformal-selection-fdp-inequality}.


\begin{theorem}\label{thm:fdp-upper-bound-cs}
    Let $G(\cdot)$ be a random function satisfying $\PP (\frac{1}{m}\sum_{j=1}^m \ind\{p_j^{\cs,\oracle} \le t\} \le G(t), \forall t\in[0,1]) \ge 1-\delta$  using methods from Section \ref{sec:upper-bound-oracle}. For selection sets of the form $\cR(t) = \{j \in [m]: p_j^\cs \le t\}$, we have
    \begin{align*}
        \PP \bigg(\fdp(\cR(t)) \le \frac{mG(t)}{1 \vee |\cR(t)|}, ~\forall ~t\in[0,1]\bigg) \ge 1-\delta.
    \end{align*}
\end{theorem}

The upper bound in Theorem~\ref{thm:fdp-upper-bound-cs} is particularly tight if the inequality in~\eqref{eq:upper_bound_cs_fd} is an equality, which is the case for the clipped scores. 



\subsection{Application to drug discovery}
\label{sec:cs-experiment}
We now provide the full details for the drug-target interaction experiment previewed in Figure \ref{fig:intro}. We evaluate our method on the drug-target interaction (DTI) prediction task following \cite{jin2023aselection}. The goal is to identify promising drug-target pairs from a large pool of candidate pairs, which can help pharmaceutical companies prioritize their drug development efforts. We use the DAVIS dataset \citep{davis2011comprehensive}, which contains $N=30,060$ drug-target pairs (samples). The feature is the structure of the drug candidate and the disease target, and the response is their binding affinity measurement. {We first randomly sample 20\% of the data to train a small neural network over $10$ epochs.} 
The remaining data is randomly split into calibration (20\%) and test (60\%) sets. We use the clipped score function with threshold $c_j$ set to the $80$th percentile of the binding affinities in the training pairs with the same target as test sample $j$. 
Here, we choose the $80$th percentile threshold rather than a higher value to encourage more discoveries.

\paragraph{Validity}
We evaluate the simultaneous validity of our FDP upper bounds across repeated experimental runs. 
As illustrated in Figure~\ref{fig:intro}a, a single realization shows that the proposed upper bound lies above the true FDP across all thresholds. 
Figure~\ref{fig:intro}b reports the difference between the FDP upper bound and the true FDP over $100$ independent runs. 
In $92$ out of $100$ runs, the upper bound remains above the true FDP for all thresholds, corresponding to an empirical coverage rate of $92\%$. 
This is consistent with the nominal confidence level $1-\delta = 0.9$.



\begin{figure}[htbp]
    \centering
    \includegraphics[width=\linewidth]{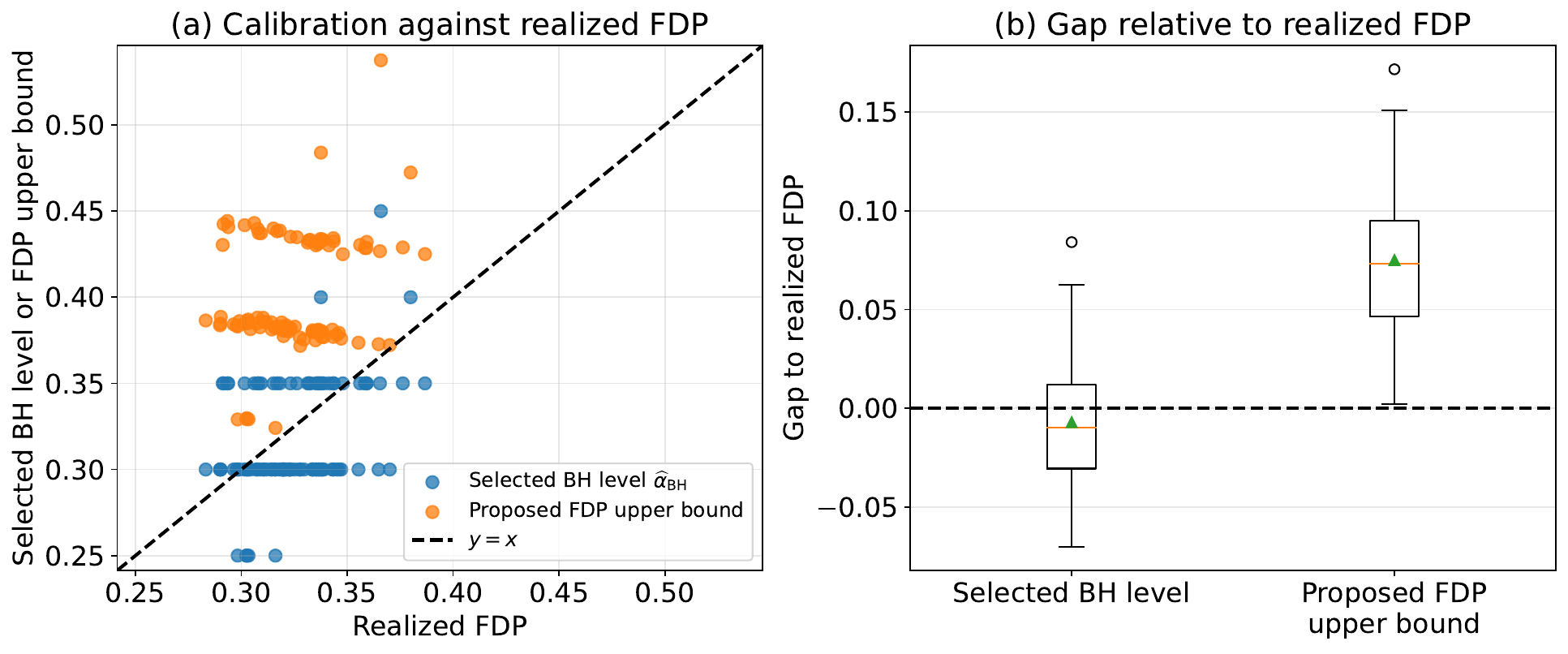}
    \caption{
    {
    Failure of post hoc BH levels as FDP certificates. Starting from $\alpha=0.05$, the analyst increases the BH nominal level in increments of $0.05$ until at least $5\%$ of the test samples are selected. Left: realized FDP versus two candidate certificates: the post-hoc BH level $\widehat\alpha_{\mathrm{BH}}$ and our proposed simultaneous FDP upper bound. Points below the diagonal fail to upper-bound the realized FDP. Right: distribution of certificate minus realized FDP across experiments. 
    }}
    \label{fig:drug_interaction_BH}
\end{figure}

\paragraph{FDP upper bound for an adaptive selection rule}
We mimic a realistic scenario in which an analyst wishes to select at least $5\%$ of the test samples. Starting from the BH procedure at level $\alpha=0.05$, the nominal level is increased by $0.05$ until the number of selections is larger than $5\%$ of the total number of test samples. This produces an adaptive selection rule that depends on the observed data.

Figure~\ref{fig:drug_interaction_BH} summarizes the behavior of the proposed FDP upper bound under this procedure. The left panel compares the selected nominal level and the proposed upper bound against the true FDP across experiments. The right panel displays the distribution of certificate minus realized FDP across experiments. Overall, the proposed upper bound consistently dominates the true FDP across experiments, whereas the post-hoc BH level underestimates the true FDP.

\section{Conclusion} \label{sec:conclusion}
In this paper, we introduce a unified framework for simultaneous inference in conformal prediction. Our method builds on constructing high-probability upper bounds for the empirical CDF of conformal uniform variables. We develop simultaneous coverage guarantees for two key applications: false coverage proportion (FCP)   in conformal prediction and false discovery proportion (FDP) in multiple testing problems with conformal p-values, such as outlier detection and conformal selection. Through extensive experiments on both synthetic and real datasets, we demonstrate that our method provides valid simultaneous coverage while remaining practically efficient.

Our framework is highly versatile, as it can be applied to many other problems  as long as we can simulate from the joint distribution of p-values. Our approach also allows one to choose different summary/test statistics to achieve various inferential goals, making it applicable beyond the specific examples presented in this work.

Our work has several limitations.
A key strength of our approach is that the proposed FDP upper bounds hold simultaneously over all thresholds; however, this generality can lead to conservative bounds when evaluated at a single threshold, such as the threshold selected by the Benjamini--Hochberg procedure (see Figure~\ref{fig:drug_interaction_BH}). Future work may develop tighter bounds for a specific data-adaptive thresholds, including the Benjamini--Hochberg threshold, and may extend the framework beyond exchangeable settings to handle covariate shift between calibration and test data.

\subsection*{Acknowledgments}

EC was supported by ONR (N00014-24-1-2305) and NIH (1R01AG08950901A1). 

\newpage
\bibliographystyle{alpha}
\bibliography{reference}

\newpage 
\appendix



\section{Adaptive issue in BH procedure}

Suppose a scientist faces a multiple testing problem with $100$ hypotheses and corresponding p-values. 
He first applies the Benjamini--Hochberg (BH) procedure at level $\alpha=0.05$, believing this to be a reasonable starting point. 
The BH procedure rejects all hypotheses with ordered p-values $p_{(i)}$ satisfying 
\[
p_{(i)} \leq \frac{i}{m}\alpha, \quad i=1,\dots,m,
\]
where $m$ is the total number of tests. 
Geometrically, this corresponds to comparing the sorted p-values against the \emph{rejection line} of slope $\alpha/m$. 

As shown in Figure \ref{fig:adaptive_example}, which depicts one realization of $100$ sorted p-values, only two points fall below the rejection line for $\alpha=0.05$, leading to just two rejections. 
Finding this too conservative, the scientist then increases the threshold to $\alpha=0.1$, which yields more rejections. 
While this post-hoc adjustment is natural in practice, it invalidates the formal FDR guarantee of the BH procedure, which requires the significance level to be fixed \emph{a priori}. 
This motivates our simultaneous FDP bounds, which remain valid uniformly across all thresholds.

\begin{figure}[H]
    \centering
    \includegraphics[width=0.6\linewidth]{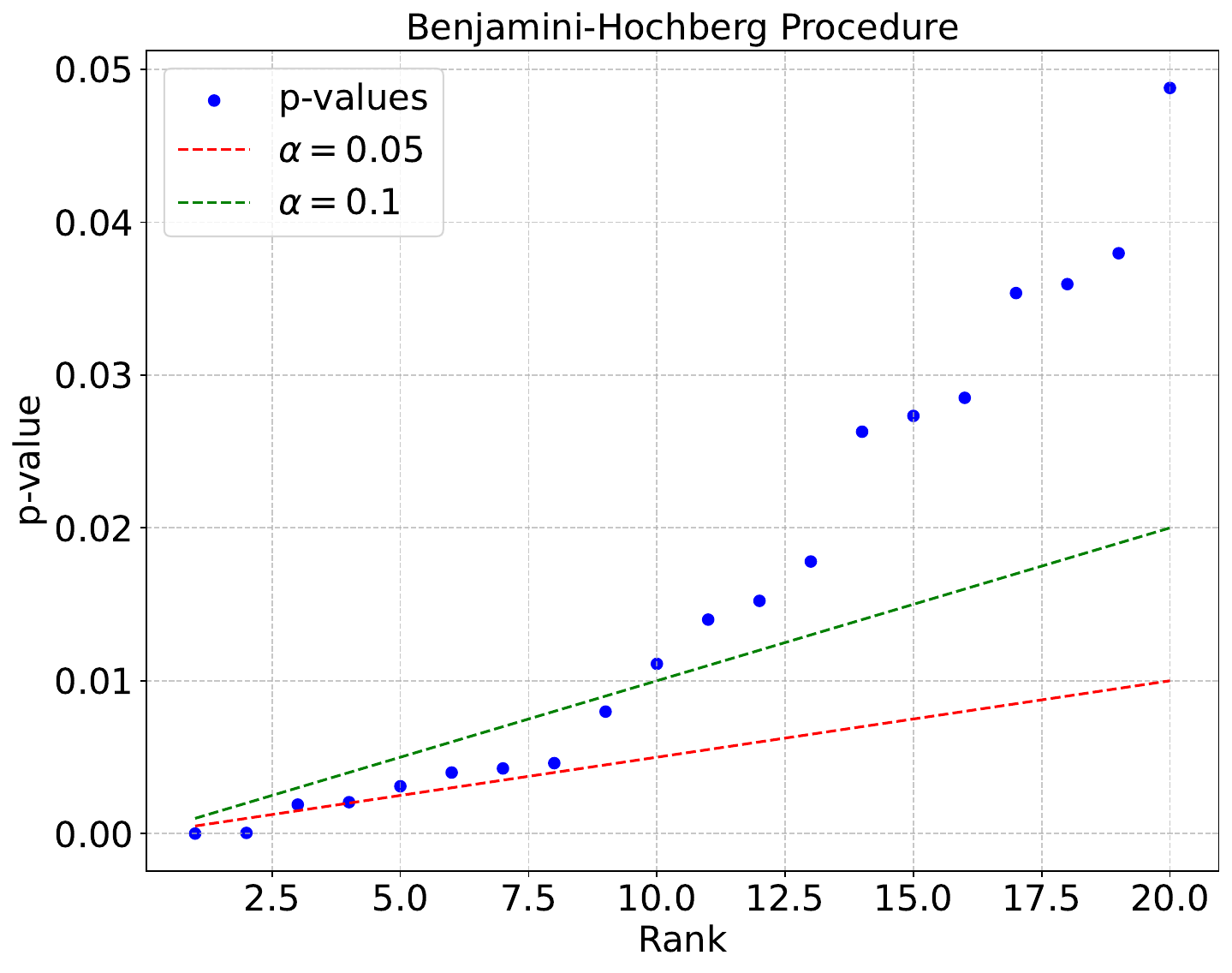}
    \caption{The risk of post hoc parameter selection.
A realization of $m=100$ p-values. The red dashed line represents the Benjamini-Hochberg rejection threshold at $\alpha=0.05$ (2 rejections). If the analyst adaptively relaxes the parameter to $\alpha=0.1$ (green dashed line) after observing the scant yield, the rigorous FDR guarantee is voided. }
    \label{fig:adaptive_example}
\end{figure}

\section{Role of the shape-parameter $\beta$}\label{app:beta}
We study the effect of the shape parameter $\beta$ in the Higher-Criticism statistic \eqref{eq:stat_thc} by constructing envelope functions for different values of $\beta$, with $(\ell,r)=(0,1)$. Figure~\ref{fig:beta} illustrates how $\beta$ influences the shape of the envelope.

For smaller values of $\beta$, the resulting envelope is approximately linear and remains closer to the ECDF in the central region of the domain (i.e., for $t$ near $0.5$). As $\beta$ increases, the envelope becomes increasingly nonlinear: it lies closer to the ECDF for small values of $t$, while being more conservative in the mid-range of $t$.

This behavior highlights a trade-off in how the envelope distributes conservativeness across the domain. If inference accuracy is primarily desired for small $t$, larger values of $\beta$ are preferable. Conversely, smaller values of $\beta$ provide a more uniform proximity to the ECDF across $t$. In practice, we find that $\beta=0.5$ offers a reasonable compromise between these two regimes.

\begin{figure}[htbp]
    \centering
    \includegraphics[width = 0.5\linewidth]{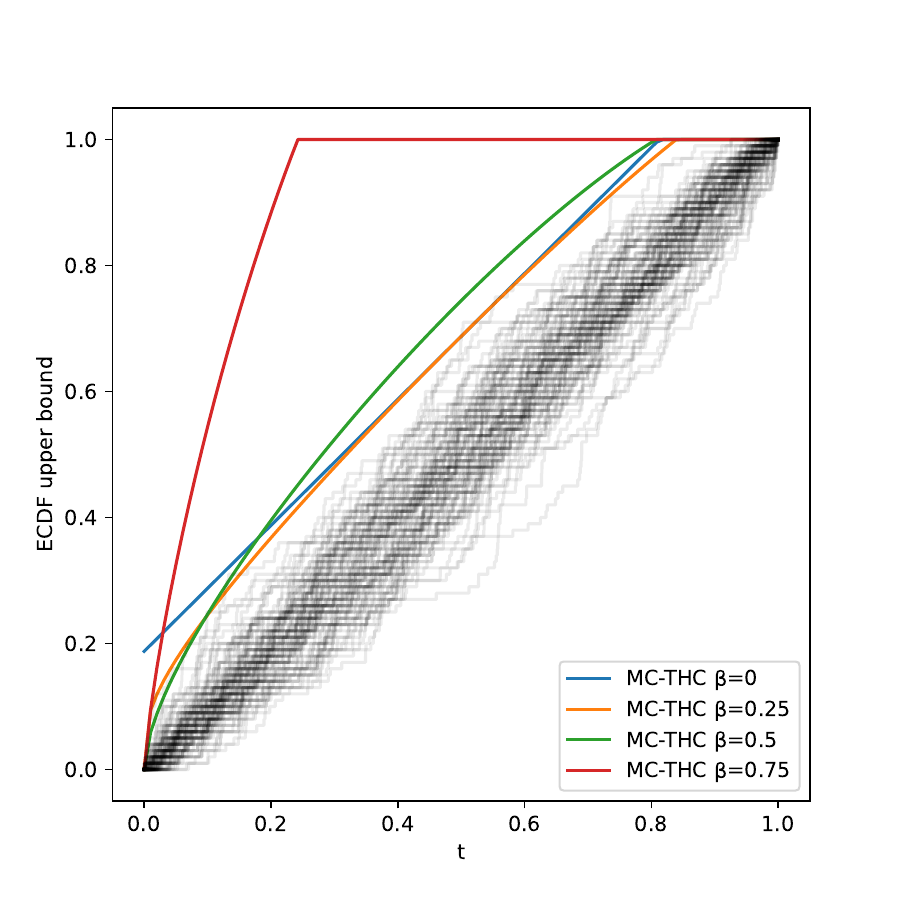}
    \caption{Upper bounds on $\hat{F}_{n,m}(t)$ ($n=m=100,$ $\delta=0.1$) constructed via Algorithm \ref{alg:monte-carlo} with $B=100$ with different shape parameter $\beta$. The gray curves represent 100 independent realizations of $\hat{F}_{n,m}$.}
    \label{fig:beta}
\end{figure}

\section{Applications to i.i.d.\ p-values}\label{app:iid-pvalues}

\subsection{Conformal vs.\ i.i.d.\ p-values: what changes and why}\label{app:comparison-iid-conformal}
Our envelope-based technique from Section~\ref{sec:upper-bound-oracle} also applies to i.i.d.\ p-values
$U_1,\dots,U_m \sim \mathrm{Unif}[0,1]$. This setting can be viewed as the ``infinite calibration'' limit of conformal
p-values: as the calibration size $n \to \infty$, the conformal p-values approach independent uniforms.
Let 
\[
\hat F_m(t) = \frac{1}{m}\sum_{j=1}^m \ind\{U_j \le t\}, \quad t \in [0,1],
\]
denote the empirical CDF of i.i.d.\ p-values. Figure~\ref{fig:iid-vs-conformal} compares $\hat F_{n,m}(t)$ from conformal
p-values with $\hat F_m(t)$ from i.i.d.\ uniforms. The contrast is clear:  
\begin{itemize}
    \item \textbf{i.i.d.\ case:} $\hat F_m(t)$ concentrates around $t$ at rate $O_\mathbb{P}(m^{-1/2})$ by central limit theorem, so the curve
    tightens to $y=x$ as $m$ grows.  
    \item \textbf{Conformal case:} with fixed $n$, calibration randomness persists regardless of $m$, so
    $\hat F_{n,m}(t)$ does not collapse to $y=x$.  
\end{itemize}

\begin{figure}[htbp]
    \centering
    \begin{subfigure}[b]{0.37\textwidth}
        \includegraphics[width=\textwidth]{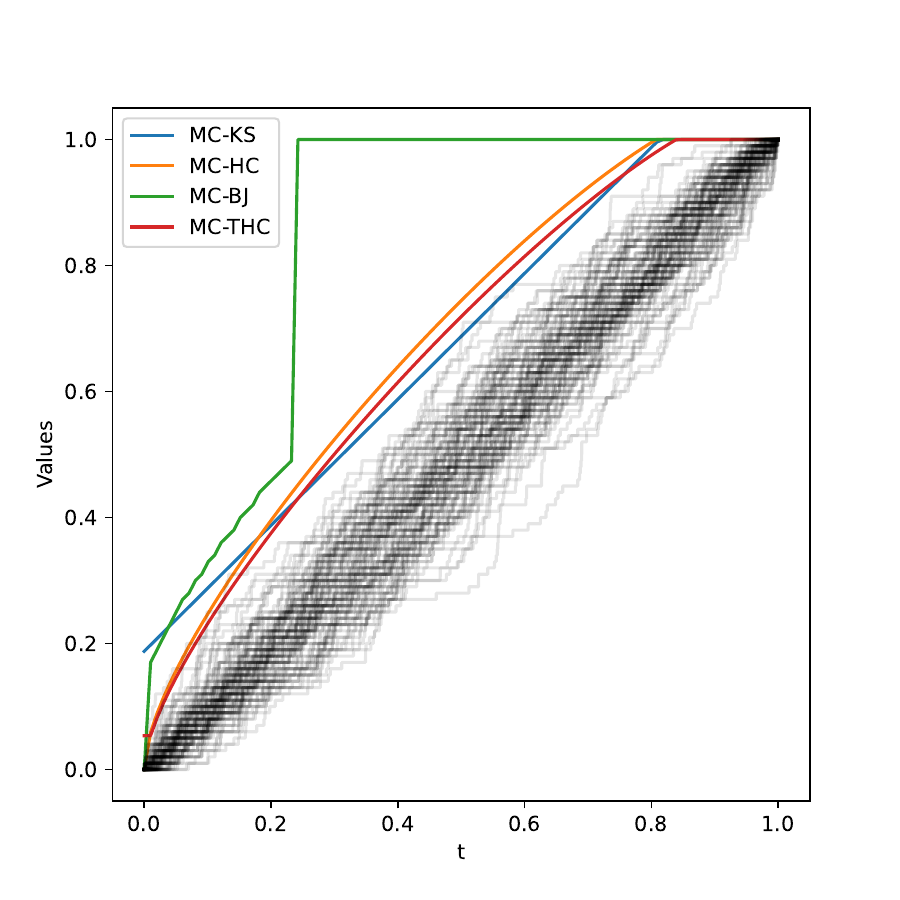}
        \caption{Conformal, $n=m=100$}
    \end{subfigure}
    \begin{subfigure}[b]{0.37\textwidth}
        \includegraphics[width=\textwidth]{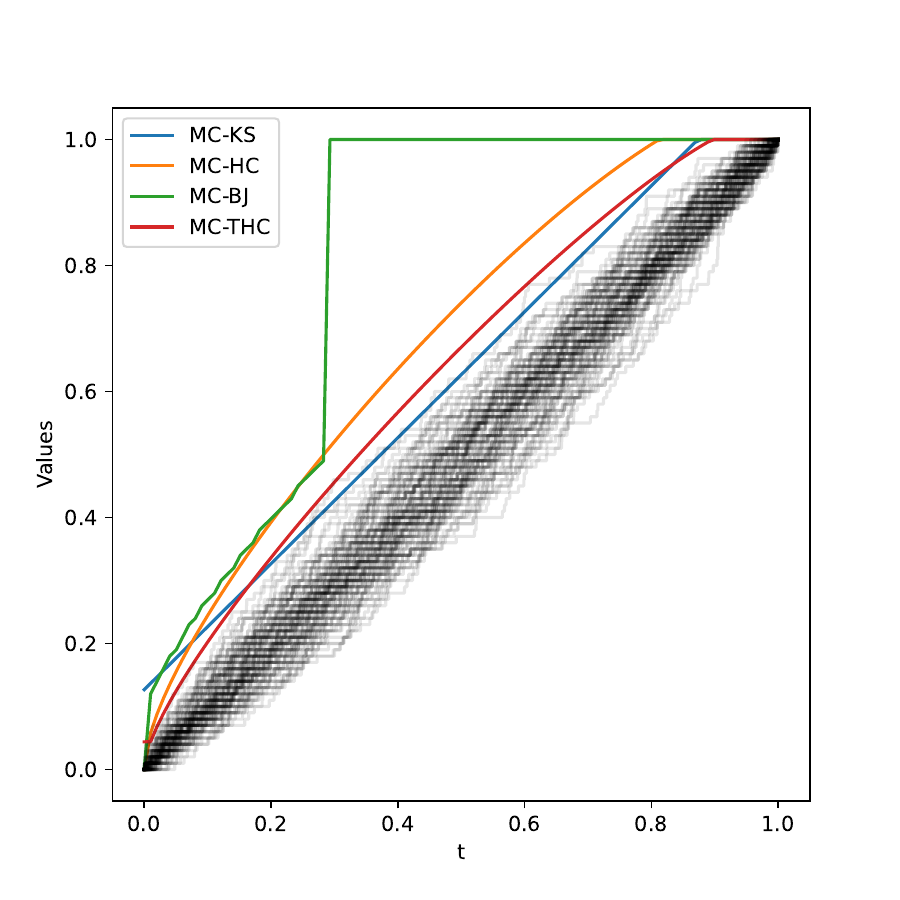}
        \caption{i.i.d., $m=100$}
    \end{subfigure}

    \begin{subfigure}[b]{0.37\textwidth}
        \includegraphics[width=\textwidth]{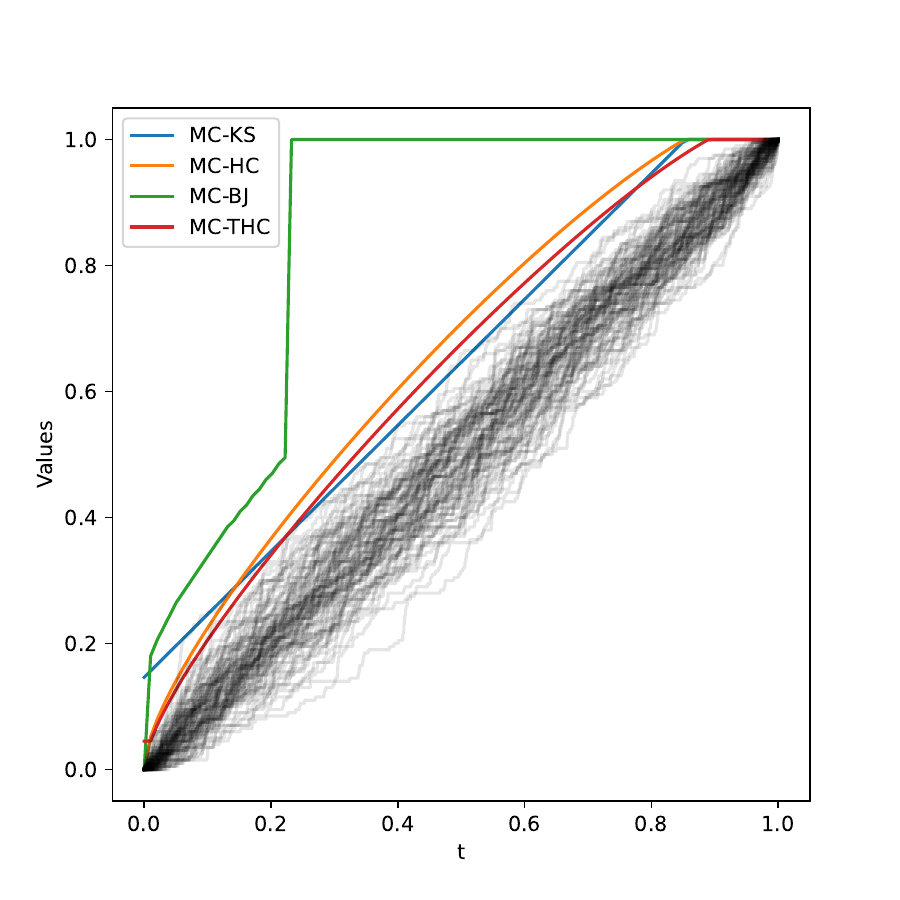}
        \caption{Conformal, $n=100,m=200$}
    \end{subfigure}
    \begin{subfigure}[b]{0.37\textwidth}
        \includegraphics[width=\textwidth]{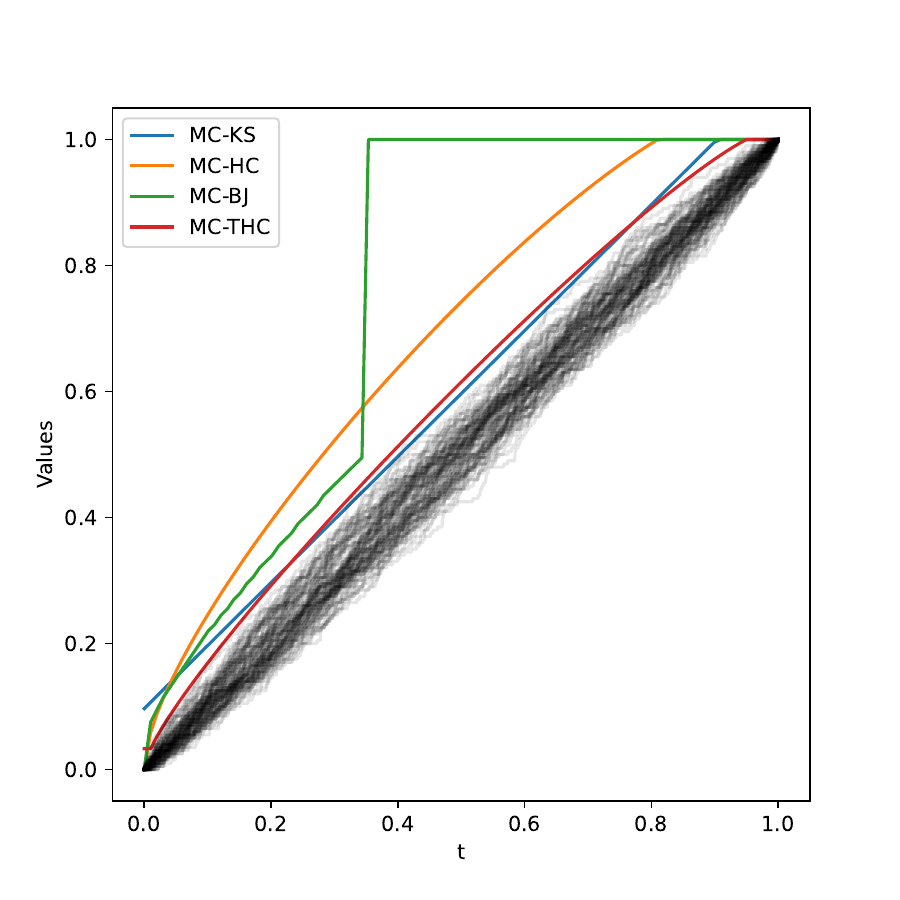}
        \caption{i.i.d., $m=200$}
    \end{subfigure}

    \begin{subfigure}[b]{0.37\textwidth}
        \includegraphics[width=\textwidth]{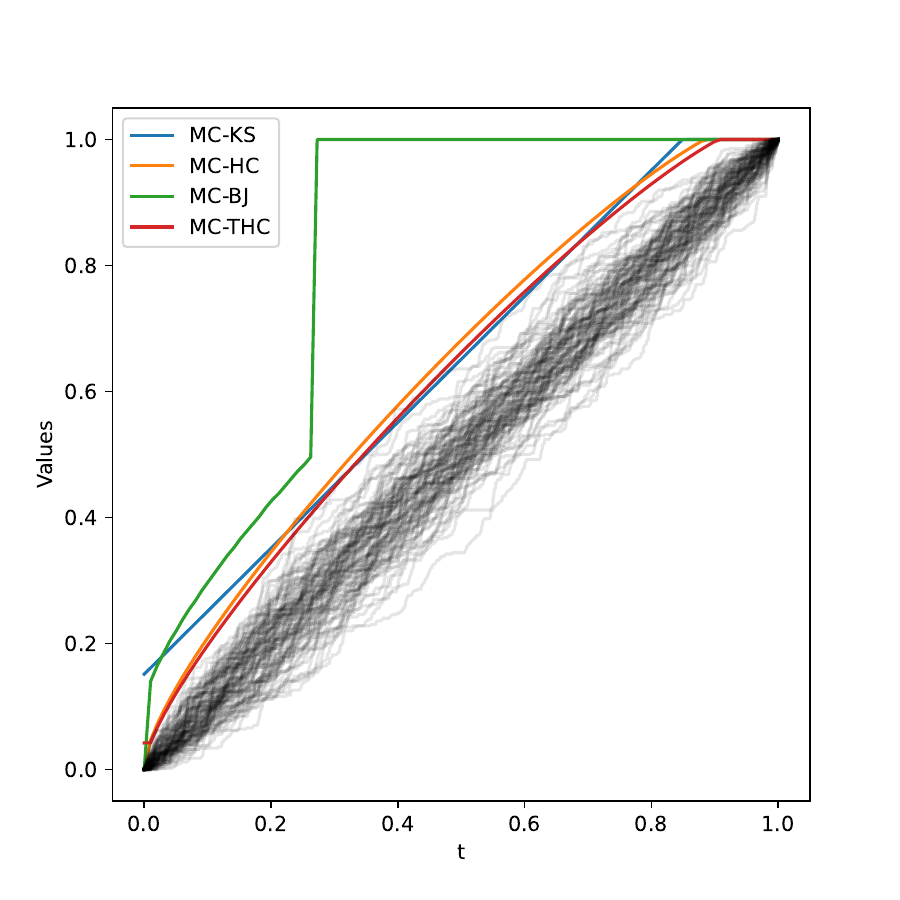}
        \caption{Conformal, $n=100,m=500$}
    \end{subfigure}
    \begin{subfigure}[b]{0.37\textwidth}
        \includegraphics[width=\textwidth]{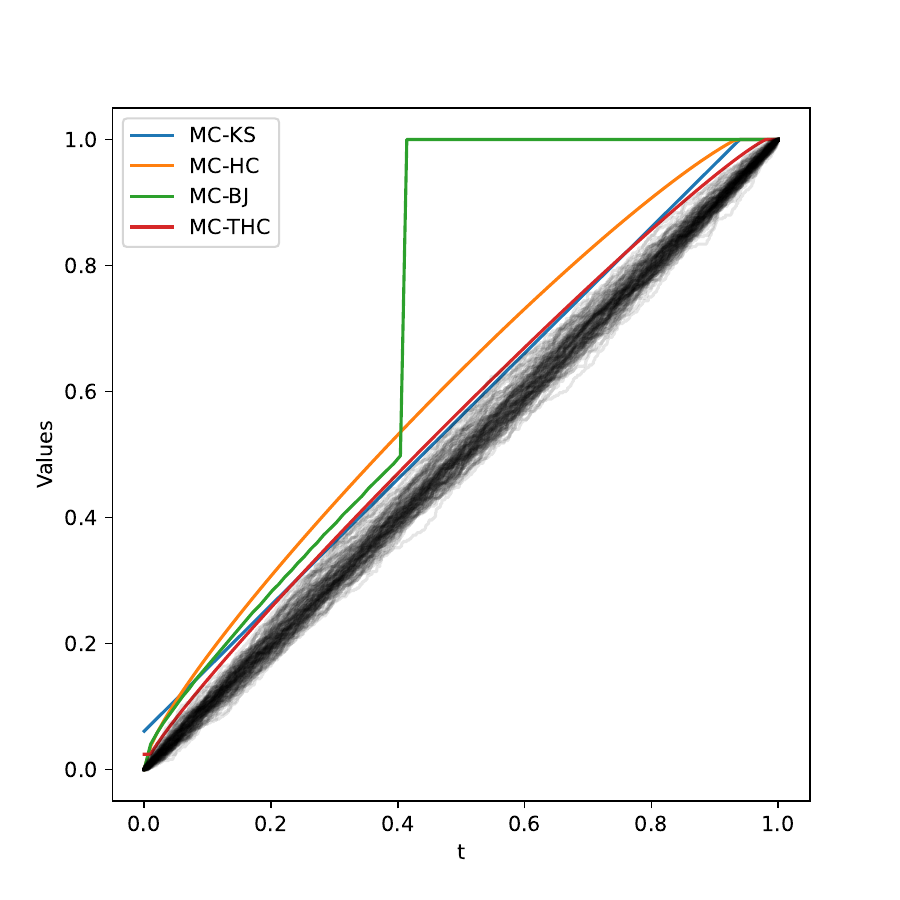}
        \caption{i.i.d., $m=500$}
    \end{subfigure}
    \caption{\textbf{ECDF envelopes: Conformal vs. i.i.d. p-values.}
Left Column: Empirical CDFs of conformal p-values with fixed calibration size $n=100$. Note that as $m$ increases, the variance does not vanish due to the persistent randomness of the finite calibration set.
Right Column: Empirical CDFs of i.i.d. uniform p-values. The distribution concentrates tightly around the diagonal $y=x$ as $m \to \infty$. This contrast highlights why specialized bounds are necessary for finite-sample conformal inference.}
    \label{fig:iid-vs-conformal}
\end{figure}

\subsection{Variance of the empirical CDF of conformal uniform variables}
\label{app:variance-ecdf}

The conformal uniform variables are generally dependent because they share the
same calibration sample. However, the variance of their empirical CDF has
the same $t(1-t)$ shape, up to a multiplicative factor.

\begin{prop}
\label{prop:var-ecdf-conformal-uniform}
Let $q_1,\ldots,q_m$ follow the conformal-uniform distribution $P_{n,m}$ in
Proposition~\ref{prop:distribution-oracle-p-values}. Define
\[
\widehat F_{n,m}(t)
=
\frac{1}{m}\sum_{j=1}^m \mathbf 1\{q_j\le t\},
\qquad t\in[0,1].
\]
Then
\[
\operatorname{Var}\{\widehat F_{n,m}(t)\}
=
c_{n,m}(t)t(1-t),
\]
where, for $t\in(0,1)$,
\[
c_{n,m}(t)
=
\frac1m+
\left(1-\frac1m\right)\rho_n(t),
\]
with
\[
\rho_n(t)
=
\frac{(n+1)(k+\gamma^2)-(k+\gamma)^2}
{(n+2)(k+\gamma)(n+1-k-\gamma)}.
\]
Here
\[
k=\lfloor (n+1)t\rfloor,\qquad
\gamma=(n+1)t-k.
\]
At the endpoints $t=0$ and $t=1$, the variance is zero.
\end{prop}

\begin{proof}
By Proposition~\ref{prop:distribution-oracle-p-values}, it suffices to work with
\[
q_j
=
\frac{\sum_{i=1}^n \mathbf 1\{T_i<T_{n+j}\}+U_j}{n+1},
\qquad j=1,\ldots,m,
\]
where $T_1,\ldots,T_{n+m}$ are i.i.d.\ $\mathrm{Unif}[0,1]$ and
$U_1,\ldots,U_m$ are i.i.d.\ $\mathrm{Unif}[0,1]$, independent of the
$T_i$'s.

Fix $t\in(0,1)$ and write
\[
a=(n+1)t,\qquad k=\lfloor a\rfloor,\qquad \gamma=a-k.
\]
Let $T_{(1)}<\cdots<T_{(n)}$ be the order statistics of
$T_1,\ldots,T_n$, and define the spacings
\[
D_0=T_{(1)},\qquad
D_\ell=T_{(\ell+1)}-T_{(\ell)},\ \ell=1,\ldots,n-1,\qquad
D_n=1-T_{(n)}.
\]
Conditional on the calibration variables $T_1,\ldots,T_n$, the indicators
\[
I_j(t)=\mathbf 1\{q_j\le t\},\qquad j=1,\ldots,m,
\]
are i.i.d.\ Bernoulli random variables with success probability
\[
H_t=\mathbb P(q_j\le t\mid T_1,\ldots,T_n).
\]
Hence
\[
\mathbb E[\widehat F_{n,m}(t)\mid T_1,\ldots,T_n]=H_t
\]
and
\[
\operatorname{Var}(\widehat F_{n,m}(t)\mid T_1,\ldots,T_n)
=
\operatorname{Var}\left(
\frac1m\sum_{j=1}^m I_j(t)
\,\middle|\, T_1,\ldots,T_n
\right)
=
\frac1m H_t(1-H_t).
\]
By the law of total variance,
\[
\operatorname{Var}\{\widehat F_{n,m}(t)\}
=
\operatorname{Var}\!\left(
\mathbb E[\widehat F_{n,m}(t)\mid T_1,\ldots,T_n]
\right)
+
\mathbb E\!\left[
\operatorname{Var}(\widehat F_{n,m}(t)\mid T_1,\ldots,T_n)
\right].
\]
Therefore,
\[
\operatorname{Var}\{\widehat F_{n,m}(t)\}
=
\operatorname{Var}(H_t)
+
\frac1m\mathbb E[H_t(1-H_t)].
\]

Since $\mathbb E[H_t]=t$, this can be rewritten as
\[
\operatorname{Var}\{\widehat F_{n,m}(t)\}
=
\frac{t(1-t)}{m}
+
\left(1-\frac1m\right)\operatorname{Var}(H_t).
\]

It remains to compute $\operatorname{Var}(H_t)$. The spacings
$(D_0,\ldots,D_n)$ follow a Dirichlet distribution with all parameters equal
to one. Using the standard covariance formula for the Dirichlet distribution
\cite[Chapter~49]{kotz2000continuous}, for any weights $w_0,\ldots,w_n$,
\[
\operatorname{Var}\left(\sum_{\ell=0}^n w_\ell D_\ell\right)
=
\frac{(n+1)\sum_{\ell=0}^n w_\ell^2
-\left(\sum_{\ell=0}^n w_\ell\right)^2}
{(n+1)^2(n+2)}.
\]
Applying this identity with
\[
w_\ell=1\quad(\ell<k),\qquad
w_k=\gamma,\qquad
w_\ell=0\quad(\ell>k),
\]
gives
\[
\operatorname{Var}(H_t)
=
\frac{(n+1)(k+\gamma^2)-(k+\gamma)^2}
{(n+1)^2(n+2)}.
\]
Since $k+\gamma=(n+1)t$, we may write
\[
\operatorname{Var}(H_t)
=
\rho_n(t)t(1-t),
\]
where
\[
\rho_n(t)
=
\frac{(n+1)(k+\gamma^2)-(k+\gamma)^2}
{(n+2)(k+\gamma)(n+1-k-\gamma)}.
\]
Substituting this expression into the total-variance decomposition yields
\[
\operatorname{Var}\{\widehat F_{n,m}(t)\}
=
\left[
\frac1m+
\left(1-\frac1m\right)\rho_n(t)
\right]t(1-t).
\]
The endpoint cases $t=0$ and $t=1$ are immediate since
$\widehat F_{n,m}(0)=0$ and $\widehat F_{n,m}(1)=1$ almost surely.
\end{proof}
\paragraph{Sensitivity of the multiplicative factor.}
The multiplicative factor $c_{n,m}(t)$ may vary with $t$, especially when the
calibration size $n$ is small. However, this variation becomes less pronounced
as $n$ grows. Indeed, Figure~\ref{fig:rho_n} plots $\rho_n(t)$ for several
values of $n$ and shows that, for moderately large calibration sizes, the
factor is nearly flat over most of the interval. This supports the use of
$(t(1-t))^\beta$ as a simple variance-adaptive normalization in the
Higher-Criticism statistic: the factor $t(1-t)$ captures the dominant
heteroskedastic shape, while the remaining dependence structure is handled by
the Monte Carlo calibration in Algorithm~\ref{alg:monte-carlo}.
\begin{figure}
    \centering
    \includegraphics[width=0.5\linewidth]{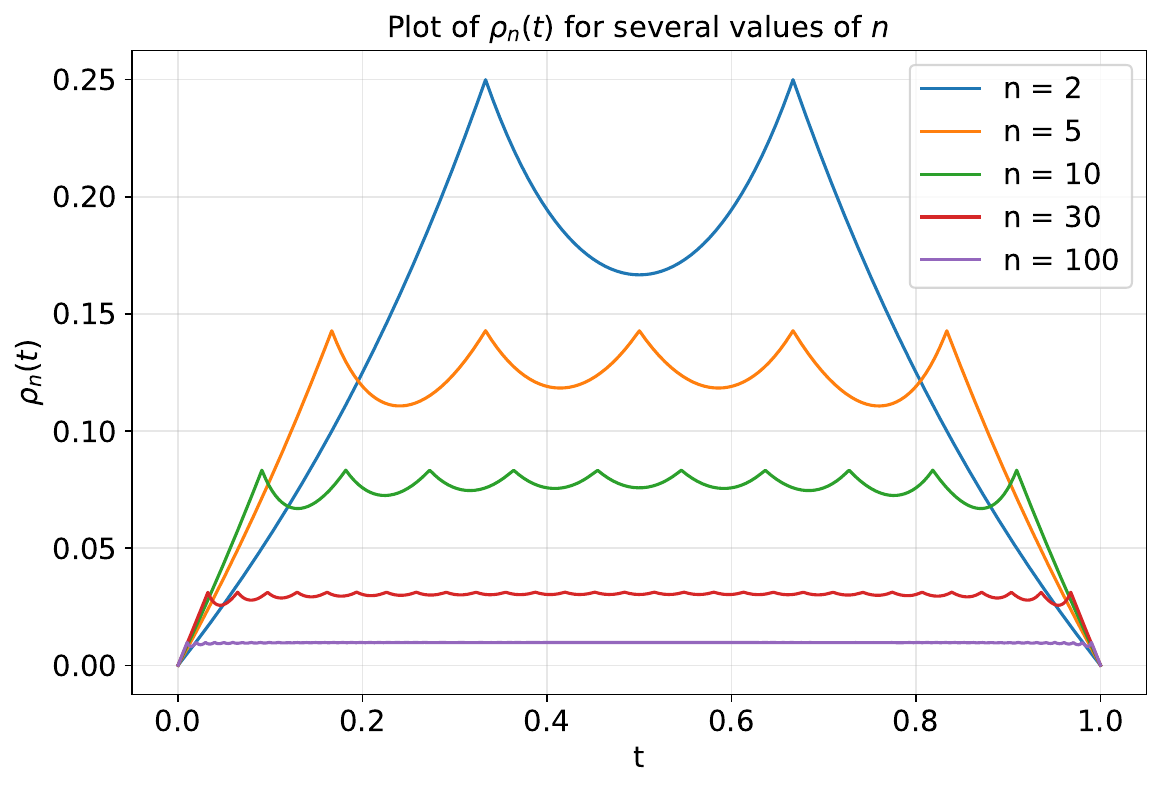}
    \caption{Plots of $\rho_n(t)$ for several values of $n$.
Since $c_{n,m}(t)=m^{-1}+(1-m^{-1})\rho_n(t)$, this also illustrates the
$t$-dependence of $c_{n,m}(t)$. The dependence is substantial only for very small
$n$ but becomes much less pronounced as $n$ grows.}
    \label{fig:rho_n}
\end{figure}
\FloatBarrier
\subsection{Constructing calibration-conditional valid p-values}\label{app:constructing-ccv-p-values}
We revisit the calibration-conditional p-values of \cite{bates2023testing} and show how our envelope construction
provides a simple and tight way to obtain them with finite-sample guarantees.

\medskip
\noindent\textbf{Calibration-conditional valid p-values}
Classical conformal (marginal) p-values are valid on average across repeated draws of the calibration data,
but may be anti-conservative when conditioning on a specific calibration set. This is unsatisfactory for
a practitioner who only has one calibration sample in hand. Calibration-Conditional Valid (CCV) p-values
remedy this issue: with probability at least $1-\delta$ over the randomness of the calibration data, each CCV
p-value is valid conditional on that calibration set. Formally, for any $t \in (0,1)$ and inlier test sample
$X_{n+1} \sim P_X$,
\[
\mathbb{P}\!\left( \mathbb{P}\!\left( \hat u^{(\mathrm{ccv})}(X_{n+1}) \le t \,\middle|\, D \right) \le t 
\ \text{for all $t$} \right) \ge 1-\delta,
\]
where $D$ denotes the calibration sample. Thus, CCV p-values provide guarantees that are specific to the user’s
data, not just “on average” across hypothetical users.

The key idea of \cite{bates2023testing} is to construct
a sequence of thresholds $(b_1,\dots,b_n)$ that uniformly lower-bound the order statistics of i.i.d.\ uniforms. 
These thresholds are then used to transform marginal p-values into calibration-conditional valid (CCV) p-values. 
The following theorem provides the precise adjustment.

\begin{theorem}[Conditional p-value adjustment]\label{thm:ccv-adjustment}
Let $U_1,\ldots,U_n \stackrel{\text{i.i.d.}}{\sim} \mathrm{Unif}[0,1]$ with order statistics
$U_{(1)} \le \cdots \le U_{(n)}$, and fix $\delta \in (0,1)$. Suppose $0 \le b_1 \le \cdots \le b_n \le 1$ satisfy
\begin{equation}\label{thm:ccv-p-values}
\mathbb{P}\big[U_{(1)} \le b_1,\ldots, U_{(n)} \le b_n\big] \ge 1-\delta.
\end{equation}
Let $b_0=0$, $b_{n+1}=1$, and define $h:[0,1]\to[0,1]$ by $h(t)=b_{\lceil (n+1)t\rceil}$. Then
$\hat u^{(\mathrm{ccv})}=h\circ \hat u^{(\mathrm{marg})}$ is calibration-conditional valid as long as $\hat u^{(\mathrm{marg})}(X_{n+1})$ is a valid p-value.
\end{theorem}


Condition \eqref{thm:ccv-p-values} is equivalent to requiring a \emph{uniform lower bound} on the empirical CDF
$\hat F_n(t)=n^{-1}\sum_{i=1}^n \ind(U_i\le t)$:
\[
\mathbb{P}\!\left(\hat F_n(b_i)\ge \frac{i}{n}\ \text{ for all } i\in[n]\right) \ge 1-\delta.
\]
Hence, constructing $b_1,\dots,b_n$ reduces to computing a lower envelope for $\hat F_n$.

\medskip
\noindent\textbf{Envelope-to-threshold construction.}
Let $L_n:[0,1]\to[0,1]$ be any nondecreasing function such that
\[
\mathbb{P}\big(\hat F_n(t)\ \ge\ L_n(t)\ \text{ for all } t\in[0,1]\big) \ge 1-\delta.
\]
Then
\[
b_i := \inf\{\, t \in [0,1]: L_n(t) \ge i/n \,\},\qquad i=1,\dots,n,
\]
satisfies \eqref{thm:ccv-p-values}. In practice, we obtain $L_n$ via Algorithm~\ref{alg:monte-carlo}, using i.i.d.\
uniforms and a \emph{lower-tail} summary statistic (e.g., a flipped KS statistic or truncated lower-tail
higher-criticism). This yields:
\begin{itemize}
\item \emph{Finite-sample} $(1-\delta)$ guarantees;
\item \emph{Adaptivity} to the shape of $\hat F_n(t)$, often yielding tighter bounds near small $t$;
\item \emph{Flexibility} to tune the choice of summary statistic for the application at hand.
\end{itemize}

Overall, this provides a straightforward, plug-and-play method to construct CCV p-values using the same
Monte Carlo algorithm developed in the main text.

We compare several methods for constructing CCV $p$-values, including the procedures introduced in \cite{bates2023testing} as well as our proposed approach. The figure plots the $\{b_i\}$ as a function of the normalized index $i/n$. Smaller values of $b_i$ correspond to less conservative thresholds and therefore greater statistical power of the resulting CCV $p$-values.

As shown in the figure, the proposed MC-THC method produces a boundary that lies closer to the empirical order statistics than classical alternatives such as DKW-type and Simes bounds, while remaining comparable to asymptotic and hybrid Monte Carlo methods. Importantly, MC-THC achieves this improved proximity while retaining an explicit finite-sample validity guarantee, leading to a favorable balance between power and theoretical rigor.

\begin{figure}
    \centering
    \includegraphics[width=0.5\linewidth]{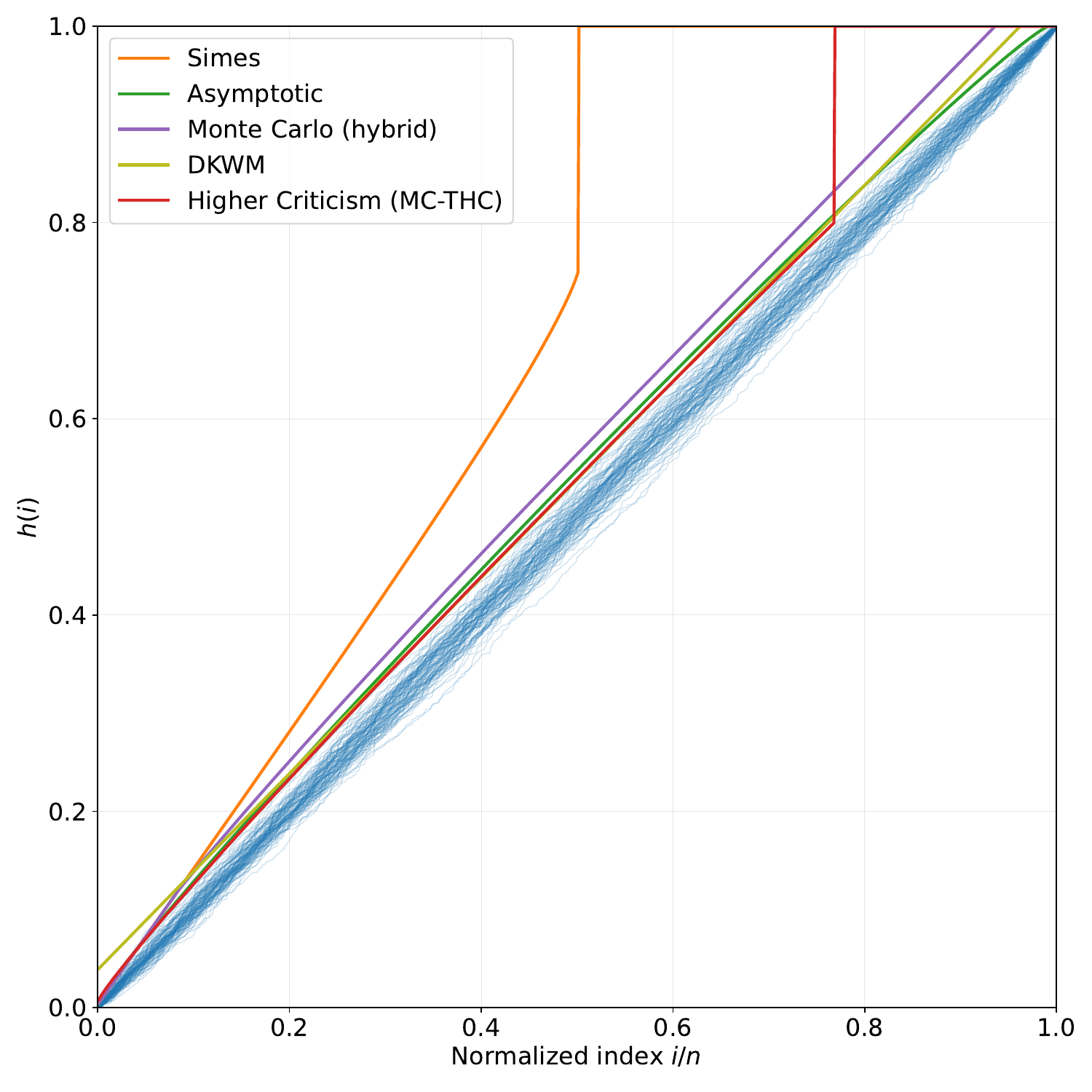}
    \caption{\textbf{Different CCV p-values adjustments.} 
The blue curves display 100 independent realizations of sorted i.i.d. uniform p-values (order statistics) with sample size $n=1000$, plotted against their normalized rank $i/n$. This setup replicates the validation framework of Figure 3 in \cite{bates2023testing}. The proposed envelope based on the Truncated Higher Criticism statistic (MC-THC, purple) is compared against other bounds in their paper. The MC-THC bound demonstrates superior tightness across a broad range of the normalized index.}
    \label{fig:order_stats}
\end{figure}
\begin{figure}
    \centering
    \includegraphics[width=0.5\linewidth]{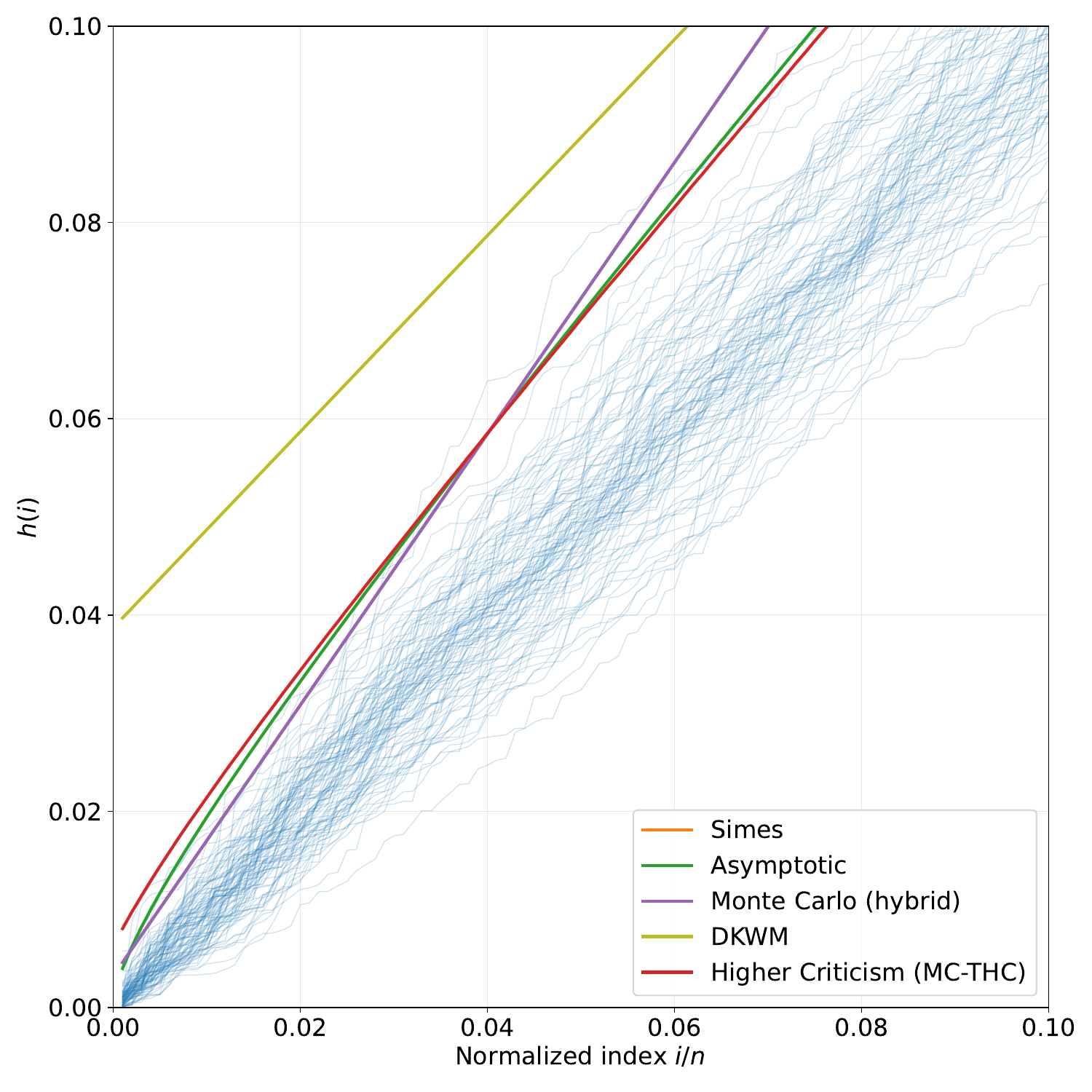}
    \caption{\textbf{Detailed view of CCV p-value adjustments in the lower tail.} 
A zoomed-in perspective of Figure \ref{fig:order_stats} focusing on the normalized index range $i/n \in [0, 0.15]$. }
    \label{fig:order_stats_zoom}
\end{figure}

\section{Proofs for Section \ref{sec:upper-bound-oracle}}\label{app:upper-bound-oracle}
First, we give the proof of Theorem \ref{thm:monte-carlo-bounds}, which is a simple consequence of exchangeability.
\begin{proof}[Proof of Theorem \ref{thm:monte-carlo-bounds}]
Let $F(t)$ be the random function specified in Algorithm \ref{alg:monte-carlo}. Let $F^{(b)}, b = 1, \cdots, B$ be $B$ independent copies of random function following the same distribution with $F$, then we have 
$$
\PP\paren{T(F) \le \quant_{1-\delta}\paren{\sum_{b=1}^B \frac{1}{B+1}\delta_{T(F^{(b)})} + \frac{1}{B+1}\delta_{T(F)}}} \ge 1-\delta.
$$
This implies with probability at least $1-\delta$,
\begin{align*}
T(F) \le &~ \quant_{1-\delta}\paren{\sum_{b=1}^B \frac{1}{B+1}\delta_{T(F^{(b)})} + \frac{1}{B+1}\delta_{T(F)}} \\
\le &~ \quant_{1-\delta}\paren{\sum_{b=1}^B \frac{1}{B+1}\delta_{T(F^{(b)})} + \frac{1}{B+1}\delta_\infty} \\
 = &~\hat{T}.
\end{align*}
Thus we have 
$$
\PP(T(F) \le \hat{T}) \ge 1-\delta.
$$
The second part of the theorem follows directly from the definition of the summary statistics in \eqref{eqn:summary-statistics}.
\end{proof}

\begin{proof}[Proof of Fact~\ref{fact:monotone}]
The monotonicity can be verified by proving that $f'(t)\leq 0$. Note that 
\begin{align*}
    f'(t) &= \frac{-(t(1-t))^\beta - \beta (x - t) (t(1-t))^{\beta-1}(1-2t)}{((t(1-t))^\beta)^2}\\
    &= -\frac{t(1-t) + \beta (x - t)(1 - 2t)}{(t(1-t))^{\beta + 1}}.
\end{align*}
The numerator $t(1-t) + \beta (x - t)(1 - 2t)$ is linear in both $\beta$ and $x$. Therefore, to prove $f'(t) \leq 0$, it suffices to verify the inequality at the extreme points $(\beta,x) \in \{(0,0), (0,1), (1,0), (1,1)\}$, which is straightforward.
\end{proof}

\section{Proofs for Section \ref{sec:outlier-detection}}\label{app:outlier-detection}
\subsection{Proofs of Lemma \ref{lem:coupling} and Theorem \ref{thm:fdp-outlier-detection-naive}}
\begin{proof}[Proof of Lemma \ref{lem:coupling}]
Let $m_0 = |\cH_0|$ be the number of nulls (inliers). By Proposition \ref{prop:distribution_pval_outlier}, the vector of null conformal
$p$-values $(p_1,\dots,p_{m_0})$ has the same distribution as the first $m_0$ coordinates of the oracle
conformal $p$-values under $\cP_{n,m_0}$ (the law in Proposition \ref{prop:distribution-oracle-p-values}).
The family $\{\cP_{n,m}\}_{m\ge 1}$ is consistent in $m$ because, under Proposition \ref{prop:distribution-oracle-p-values},
each $p_j$ depends only on $(T_1,\dots,T_n,T_{n+j},U_j)$ and is thus unaffected by the presence or
absence of other test samples. Hence, the marginal law of the first $m_0$ coordinates of $\cP_{n,m}$
equals $\cP_{n,m_0}$.
Fix $m$ and consider $\cP_{n,m}$. Let $\mathcal K(\,\cdot\,|\,\cdot\,)$ be the conditional
distribution of the remaining coordinates $((p_{m_0+1},\dots,p_m))$ given the first $m_0$ coordinates
under $\cP_{n,m}$. On the probability space carrying the given $(p_1,\dots,p_{m_0})$, define
$(p^*_{m_0+1},\dots,p^*_m)\sim \mathcal K(\,\cdot\,|\,p_1,\dots,p_{m_0})$, independently of all else.
By construction, the concatenated vector
\[
   (p_1,\dots,p_{m_0},p^*_{m_0+1},\dots,p^*_m)
\]
has joint law $\cP_{n,m}$, proving the claim.
\end{proof}

\begin{proof}[Proof of Theorem \ref{thm:fdp-outlier-detection-naive}]
Let $R(t)=\{j\in[m]:\,p_j\le t\}$ and write $V(t)=|R(t)\cap \cH_0|=\sum_{j\in \cH_0}\ind\{p_j\le t\}$.
By Lemma \ref{lem:coupling}, there exist $(p_1^*,\dots,p_m^*)\sim\cP_{n,m}$ such that
\[
   V(t)=\sum_{j\in \cH_0}\ind\{p_j\le t\}
       \le \sum_{j=1}^m \ind\{p_j^*\le t\}.
\]
On the event
\(
  \frac{1}{m}\sum_{j=1}^m \ind\{p_j^*\le t\}\le G(t)\ \text{for all }t\in[0,1],
\)
we have $\sum_{j=1}^m \ind\{p_j^*\le t\}\le m\,G(t)$ for all $t$, hence
\[
   \mathrm{FDP}(t)
   = \frac{V(t)}{1\vee |R(t)|}
   \le \frac{m\,G(t)}{1\vee \sum_{j=1}^m \ind\{p_j\le t\}}
   \qquad\text{for all }t\in[0,1].
\]
\end{proof}

\subsection{Proof of Proposition \ref{prop:tighten} and Theorem \ref{thm:fdp-upper-bound-outlier-detection}}
\begin{proof}[Proof of Proposition \ref{prop:tighten}]
Let $V(t)=|R(t)\cap \cN|$ and suppose $\PP\big(V(s)\le B(s)\ \forall s\in(0,1)\big)\ge 1-\delta$.
For any $s\le t$,
\[
   V(t) \le V(s) + \big(|R(t)| - |R(s)|\big)
   \le B(s) + |R(t)| - |R(s)|,
\]
hence
\[
   V(t) \le |R(t)| - \big(|R(s)|-B(s)\big).
\]
Taking the supremum over $s\le t$ and using $V(t)\le R(t)$ yields
\[
   V(t) \le |R(t)| - \sup_{s\le t}\big(|R(s)|-B(s)\big)_{+} \;=: \; B^\star(t).
\]
Therefore $\PP\big(V(t)\le B^\star(t)\ \forall t\big)\ge 1-\delta$, and dividing by $1\vee |R(t)|$ gives the
stated FDP bound.
For the computable form, note that both $|R(s)|$ and \textbf{(typically)} $B(s)$ are right–continuous step functions
that only change when $s$ passes an observed $p$-value. Thus the supremum over $s\le t$ is attained at
some $s=p_j\le t$, and a short rearrangement yields
\[
   B^\star(t)=\min_{j:\,p_j\le t}\Big\{ B(p_j) + |R(p_{\le t})| - |R(p_j)| \Big\}.
\]
\end{proof}

\begin{proof}[Proof of Theorem \ref{thm:fdp-upper-bound-outlier-detection}]

(i) Let $r=|\cH_0|$. By the assumed guarantee for $G_r(\cdot)$ and Lemma \ref{lem:coupling}, with probability at least $1-\delta$,
\[
   \sum_{j\in \cH_0}\ind\{p_j\le t\}
   \le \sum_{j=1}^{r}\ind\{p_j^*\le t\}
   \le G_r(t)\qquad\forall t.
\]
Equivalently,
\(
   \sum_{j\in \cH_0}\ind\{p_j> t\} \ge r - G_r(t)
\),
and hence
\[
   \sum_{j=1}^m \ind\{p_j> t\}
   \;\ge\; r - G_r(t)\qquad\forall t.
\]
By the definition of $\widehat m_0$ as the largest $r$ satisfying the same system of inequalities, it follows
that $\widehat m_0\ge |\cH_0|$ with probability at least $1-\delta$.

(ii) Let $r=|H_0|$ and define the event
\[
\mathcal E_r \;=\; \Bigl\{\sum_{j=1}^{r}\ind\{p_j^\ast \le t\}\le G_r(t)\ \text{for all } t\in[0,1]\Bigr\}.
\]
By the theorem's assumption applied at $r=|H_0|$, we have $\mathbb P(\mathcal E_r)\ge 1-\delta$.
On $\mathcal E_r$, Lemma~\ref{lem:coupling}  yields the coordinate-wise coupling
\[
\sum_{j\in H_0}\ind\{p_j\le t\}
\;\le\;
\sum_{j=1}^{r}\ind\{p_j^\ast\le t\}
\;\le\;
G_r(t)\qquad\forall t\in[0,1].
\]
Still on $\mathcal E_r$, part (i) shows $\widehat m_0\ge r$, hence
\[
B(t)\;=\;\sup_{k\le \widehat m_0}G_k(t)\;\ge\;G_r(t)\quad\text{for all }t.
\]
Combining the previous displays, on $\mathcal E_r$ we have $|R(t)\cap H_0|\le B(t)$ for all $t$, so
\[
\mathbb P\bigl(|R(t)\cap H_0|\le B(t)\ \forall t\bigr)\ \ge\ 1-\delta.
\]
Finally, applying Proposition~\ref{prop:tighten} with this simultaneous upper bound $B(t)$ gives
\[
\mathbb P\Bigl(\mathrm{FDP}(R(t))\le \frac{B^\ast(t)}{1\vee |R(t)|}\ \text{for all }t\Bigr)\ \ge\ 1-\delta,
\]
where $B^\star(t)=\min_{j:p_j\le t}\{B(p_j)+|R(p_{\le t})|-|R(p_j)|\}$.
\end{proof}

\section{Proofs for Section \ref{sec:conformal-selection}\label{app:conformal-selection}
}
\begin{proof}[Proof of Lemma \ref{lem:conformal-selection-fdp-inequality}]
From the definition of the selection set $\cR(t)$, we have
\begin{equation*}
   \begin{aligned}
    |\cR(t) \cap \cH_0| &= \sum_{j=1}^m \ind\{p_j \le t, Y_{n+j}
    \le c_{n+j}\}\\
    &\overset{}{\le} \sum_{j=1}^m \ind\{p_j \le t, V(X_{n+j}, Y_{n+j}, c_{n+j})\le V(X_{n+j},c_{n+j}, c_{n+j})\}\\
    &= \sum_{j=1}^m \ind\{p_j \le t, p_j^* \le p_j\}\\
    &\le \sum_{j=1}^m \ind\{p_j^* \le t\}.
\end{aligned}
\end{equation*}
Here, the first inequality uses the fact that the nonconformity score $V$ is monotone. 

On the other hand, if the score function satisfies
the conditions in the lemma, then if $p_j^* \le t$, we must have $Y_{n+j} \le c_{n+j}$ because otherwise $p_j \ge \frac{1}{n+1}\sum_{i=1}^n \ind\set{V(X_i,Y_i,c_i) < V(X_{n+j},Y_{n+j},c_{n+j})} \ge \frac{1}{n+1}\sum_{i=1}^n \ind\set{Y_i\le c_i} = t $. So we have $$\ind\set{p_j^* \le t} = \ind\set{p_j^* \le t, Y_{n+j} \le c_{n+j}} = \ind\set{p_j^* \le t, p_j^* = p_j} = \ind\{p_j\le t, Y_{n+j}\le c_{n+j}\}.$$
\end{proof}

\begin{proof}[Proof of Theorem \ref{thm:fdp-upper-bound-cs}]
Let
\[
\mathcal E\;=\;\Bigl\{\frac{1}{m}\sum_{j=1}^{m}\ind\{p_j^\ast\le t\}\le G(t)\ \text{for all }t\in[0,1]\Bigr\},
\]
so $\mathbb P(\mathcal E)\ge 1-\delta$ by hypothesis. On $\mathcal E$, we have
\[
\sum_{j=1}^{m}\ind\{p_j^\ast\le t\}\ \le\ m\,G(t)\qquad\forall t.
\]
For the conformal selection set $R(t)=\{j: p^{\mathrm{CS}}_j\le t\}$, Lemma~4.3 gives
\[
|R(t)\cap H_0|\ \le\ \sum_{j=1}^{m}\ind\{p_j^\ast\le t\}\ \le\ m\,G(t)\qquad\forall t.
\]
Dividing by $1\vee |R(t)|$ yields
\[
\mathrm{FDP}(R(t))\ \le\ \frac{m\,G(t)}{1\vee |R(t)|}\qquad\forall t
\]
on $\mathcal E$, and therefore
\[
\mathbb P\Bigl(\mathrm{FDP}(R(t))\le \frac{m\,G(t)}{1\vee |R(t)|}\ \text{for all }t\Bigr)\ \ge\ 1-\delta.\qedhere
\]
\end{proof}

\end{document}